\documentclass{article}

\usepackage{arxiv}

\usepackage[utf8]{inputenc} % allow utf-8 input
\usepackage[T1]{fontenc}    % use 8-bit T1 fonts
\usepackage{hyperref}       % hyperlinks
\usepackage{url}            % simple URL typesetting
\usepackage{booktabs}       % professional-quality tables
\usepackage{nicefrac}       % compact symbols for 1/2, etc.
\usepackage{microtype}      % microtypography
\usepackage{graphicx}
\usepackage{natbib}
\usepackage[table]{xcolor}

\usepackage{wrapfig}
\usepackage{subcaption}
\usepackage{lmodern}
\usepackage{anyfontsize}
\usepackage{multirow}%
\usepackage{amsmath,amssymb,amsfonts}
\usepackage{amsthm}%
\usepackage[title]{appendix}%

\newtheorem{theorem}{Theorem}

\newtheorem{lemma}[theorem]{Lemma}
\newtheorem{example}{Example}%
\newtheorem{remark}{Remark}%
\newtheorem{definition}{Definition}%

\usepackage{collcell}
\usepackage{siunitx}
\usepackage{tikz}

\newcommand\Heat[1]{
    \pgfmathsetmacro\result{min(#1/3.83,1)^2*100}
      \edef\HeatCell{\noexpand\cellcolor{white!\result!cyan}}%
    \HeatCell$#1$%
}

\newcolumntype{H}{>{\collectcell\Heat}r<{\endcollectcell}}

\newcommand\HeatL[1]{
  \pgfmathsetmacro\result{min(#1/3.74,1)^(2)*100}
      \edef\HeatCell{\noexpand\cellcolor{white!\result!cyan}}%
    \HeatCell$#1$%
}
\newcolumntype{L}{>{\collectcell\HeatL}r<{\endcollectcell}}

\newcommand*{\doi}[1]{\href{https://doi.org/\detokenize{#1}}{doi: \detokenize{#1}}}

\title{Sparse Data-Driven Random Projection in Regression for High-Dimensional Data}

\date{March 7, 2024}	% Here you can change the date presented in the paper title
\author{\href{https://orcid.org/0000-0003-0893-3190}{\includegraphics[scale=0.06]{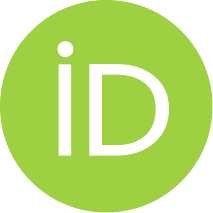}\hspace{1mm}Roman~Parzer} \\
	TU Wien\\
	Vienna, Austria \\
	\texttt{roman.parzer@tuwien.ac.at} \\
  \AND
	\href{https://orcid.org/0000-0002-8014-4682}{\includegraphics[scale=0.06]{./PlotsPaper2024/orcid.pdf}\hspace{1mm}Peter~Filzmoser} \\
	TU Wien\\
	Vienna, Austria \\
	\And
	\href{https://orcid.org/0000-0002-9613-7604}{\includegraphics[scale=0.06]{./PlotsPaper2024/orcid.pdf}\hspace{1mm}Laura~Vana-G\"ur} \\
	TU Wien\\
	Vienna, Austria
}

\renewcommand{\shorttitle}{Sparse Data-Driven Random Projection in Regression}

\hypersetup{
pdftitle={Sparse Data-Driven Random Projection in Regression for High-Dimensional Data},
pdfsubject={stat.ME},
pdfauthor={Roman~Parzer, Peter~Filzmoser, Laura~Vana-G\"ur},
pdfkeywords={Dimension Reduction, Ensemble Learning, High-Dimensional Linear Regression, Prediction, Probabilistic Screening},
}

\begin{document}
\maketitle

\begin{abstract}
	We examine the linear regression problem in a challenging high-dimensional setting with correlated predictors where the vector of coefficients can vary from sparse to dense. 
In this setting, we propose a combination of  
probabilistic variable screening with random projection tools as a viable approach. More specifically, we introduce a new data-driven random projection tailored to the problem at hand and derive a theoretical bound on the gain in expected prediction error over conventional random projections. The variables to enter the projection are screened by accounting for predictor correlation. To reduce the dependence on fine-tuning choices, we aggregate over an ensemble of linear models. A thresholding parameter is introduced to obtain a higher degree of sparsity. Both this parameter and the number of models in the ensemble can be chosen by cross-validation.

In extensive simulations, we compare the proposed method with other random projection tools and with classical sparse and dense methods and show that it is competitive in terms of prediction across a variety of scenarios with different sparsity and predictor covariance settings. We also show that the method with cross-validation is able to rank the variables satisfactorily. 
Finally, we showcase the method on two real data applications. 
\end{abstract}

% keywords can be removed

\section{Introduction}\label{sec:intro}

The recent advances in technology have allowed
more and more quantities to be tracked and stored, which has led to a huge increase in the 
amount of data, making available datasets more complex and larger than ever, both in dimension and size. 
We consider a standard linear regression setting, where the response variable is given by
\begin{align}\label{eqn:datastr}
  y_i = \mu + x_i'\beta + \varepsilon_i,\quad i=1,\dots,n.
\end{align}
Here, $n$ is the number of observations, $\mu$ is a deterministic intercept, the $x_i$s are iid observations of 
$p$-dimensional covariates or predictors with common covariance matrix $\Sigma\in\mathbb{R}^{p\times p}$, $\beta=(\beta_1,\dots,\beta_p)' \in \mathbb{R}^p$ is an unknown parameter vector and 
the $\varepsilon_i$s are iid error terms with $\mathbb{E}[\varepsilon_i]=0$ and constant $\text{Var}(\varepsilon_i)=\sigma^2$ 
independent from the $x_i$s. We are interested in studying the case where $p>n$ or even $p\gg n$.

A modern tool used for tackling the curse of dimensionality in statistics and machine learning is random projection, a method that first came up in the area of compression to speed up computation and save storage and which linearly maps for a 
set of points in high dimensions into a much lower-dimensional space while approximately preserving pairwise distances.
Possible applications are low-rank approximations \citep{Clarkson2013LowRankApprox}, data reduction for high $n$ \citep[e.g.,][]{Geppert2015RPforBayReg,ahfock2021statistical}, or data privacy \citep[e.g.,][]{Zhou2007CompressedRegression}.
It has also been applied to project the predictors to a random lower dimensional space in linear regression models 
to obtain predictive models, see e.g \cite{Maillard2009CompressedLS}, \cite{BayComprRegr}, \cite{Dunson2020TargRandProj}.
\citet{Thanei2017RPforHDR} discuss the application of random projection for column-wise compression in linear regression problems and give
an overview of theoretical guarantees on generalization error.
Although several works include desirable theoretical 
asymptotic properties for random projection techniques, it might be beneficial in practice to combine information of multiple such reductions 
to account for the variance introduced through the additional sources of randomness \citep{Thanei2017RPforHDR,BayComprRegr}. Moreover, for very large $p$, random projection
can suffer from overfitting, as too many irrelevant predictors are being considered for prediction purposes \citep{Dunson2020TargRandProj}.

In this paper, we propose a new random projection designed for dimension reduction in linear regression, which takes the variables' effect on the response into consideration while also accounting for the correlation among the predictors. A theoretical bound on the expected gain in prediction error compared to a conventional random projection is provided.
Using this random projection matrix, we do not project all predictors onto the lower dimensional space, but rather a subset by using a randomized screening step  \citep[similar to][]{Dunson2020TargRandProj}. In this step, predictors are selected with a probability proportional to their effect on the response conditional on the other predictors. To reduce the variance introduced through the two randomness sources, we build an ensemble of screened and projected linear models.
To further add a higher degree of sparsity into the models we introduce a thresholding parameter.
The number of models in the ensemble and the thresholding parameter can be chosen by cross-validation.

The idea of variable screening, where a subset of variables is rapidly selected for further analysis based on some utility measure for the predictors, and the rest of the variables are disregarded, is commonly employed for reducing the dimensionality in a regression setting. 
A seminal work in this field is sure independence screening (SIS) by \citet{Fan2007SISforUHD}, where the variables with highest
absolute correlation to the response are selected while preserving the true model with an overwhelming probability. The SIS approach works under the assumption that marginal correlations for the important variables must be bounded
away from zero, which does not always hold in practice.
Other approaches which aim to relax this assumption include forward regression screening 
\citep{Wang2009_extrCor_SuperMData} and high-dimensional ordinary least squares projection \citep[HOLP;][]{Wang2015HOLP}, which is the screening estimator that we employ in our method due to its desirable properties (as described in Section~\ref{sec:method}).

Most of the literature dealing with high-dimensional linear regression setting imposes certain sparsity assumptions on the regression coefficient $\beta$ \citep[see, e.g.,][]{Fan2010OvVarSelHD}.
It is very unlikely to obtain theoretical guarantees without additional assumptions.
For example, \cite{wainwright2019HDS_nonasy} show that there is no consistent estimator when $p/n$ is bounded away from $0$ for general $\beta$. If one assumes that $\beta$ is sparse, model selection procedures based on penalized regression such as elastic net \citep{Zou2005ElasticNet}, adaptive LASSO \citep{Zou2006AdLASSO} and many others
have been proposed to tackle the high-dimensionality of the problem. For a review of sparse linear regression in high-dimensional settings see \cite{Wang2019HDRinPracticeStudy}. In the Bayesian framework, the most common approach to inducing sparsity is through different types of shrinkage priors  \cite[see][for a discussion on the ``sparseness'' of different shrinkage priors]{gruber2023forecasting}. However, in very high dimensions, the computations required in the Bayesian models can become prohibitive \citep[see][for a detailed account on the computational complexity of
Markov chain Monte Carlo methods for high-dimensional Bayesian linear regression under sparsity assumption]{yang2016}.
%Although many classical methods, like partial least squares (PLS), principal component regression (PCR), could still be applied in the investigated setting, \citet[Chapter 18]{hastie2009elements} mention that such methods \rom{ designed for $p<n$  might behave differently in high dimensions.} 
% \rom{Cite \cite{Wang2019HDRinPracticeStudy} High dim regression in practice}

On the other hand, classical methods that do not rely on the sparsity assumption include principal component regression (PCR), which relies on the rather restrictive assumption that the response depends only on a few principal components, and  partial least squares (PLS), 
which has been shown to achieve its best asymptotic
behavior in dense settings, i.e., in regressions where many predictors contribute information about the response \citep{cook2019partial}. 
% More recently, \cite{SilinFan2022HDlinReg} proposed an approach for a non-sparse setting by only relying on the assumption of decay of eigenvalues of the covariance matrix of the data.

In practice it is often not clear what is the degree of sparsity in $\beta$ that one should expect for a given application. Therefore, we are interested in scenarios where the number of active variables 
$a=|\{j:\beta_j\neq0,1\leq j\leq p\}|$ can be larger than $n$ up to a fraction of $p$. 
In our approach, the random projection step puts no assumptions on the sparsity of $\beta$, while the screening and the thresholding procedure do introduce some degree of sparsity in the model. The flexibility in the model should make it adequate to accommodate different degrees of sparsity in the true regression coefficients. 
In a broad simulation study across six different covariance structures and three different levels of sparsity, we benchmark this new approach against an extensive collection of existing methods and show that it provides the best aggregated performance over all scenarios when looking at the ranks of prediction ability. We also show that the proposed method using cross-validation is competitive in terms of variable ranking.

% ----------- Structure ------------ %
The paper is organized as follows.
Section~\ref{sec:method} introduces the methodology.
% with one section dedicated to variable screening and one to the new random projection,
% before we combine them in one method. 
An extensive simulation study is presented in Section~\ref{sec:simulation}.
Section~\ref{sec:data} illustrates the proposed method on two real-world datasets 
% (rat eye gene expression and angles of face images), 
and Section~\ref{sec:conclusion} concludes.

\section{Methods}\label{sec:method}

 In this section, we first introduce the concept of variable screening in Section~\ref{sec:method:Scr} and motivate the use of HOLP over SIS for this purpose.
 In Section~\ref{sec:method:RP}, we present conventional random projections before proposing our random projection tailored to dimension reduction for linear regression
 and giving a theoretical bound on the performance gain in expected prediction error. Finally, we discuss how to combine these two concepts in Section~\ref{sec:method:ScrRp} and 
 propose our own algorithm in Section~\ref{sec:method:SPAR}.

% ----------- Notation ------------ %
The following notation is used throughout the rest of this paper.
For any integer $n\in\mathbb{N}$, $[n]$ denotes the set $\{1,\dots,n\}$, $I_n\in\mathbb{R}^{n\times n}$ is 
the $n$-dimensional identity matrix and $1_n\in\mathbb{R}^n$ is an $n$-dimensional vector of ones. 
From model \eqref{eqn:datastr}, we let $X\in\mathbb{R}^{n\times p}$ be the matrix of centered predictors with rows $\{x_i-\bar x: i\in[n]\}$ and
$y=(y_1-\bar y,\dots,y_n-\bar y)'\in\mathbb{R}^n$ the centered response vector,
where $\bar x = \frac{1}{n}\sum_{i=1}^nx_i \in\mathbb{R}^p,\bar y = \frac{1}{n}\sum_{i=1}^ny_i \in\mathbb{R}$.
Furthermore, $X_{.j}=(X_{1j},\dots,X_{nj})'\in\mathbb{R}^n$ denotes $j$-th columns of $X$.

\subsection{Variable Screening}\label{sec:method:Scr}

The general idea of variable screening is to select a (small) subset of variables, based on some marginal utility measure for the predictors, and disregard the rest for further analysis.
In their seminal work, \citet{Fan2007SISforUHD} propose to use the vector of marginal empirical correlations 
$w=(w_1,\ldots ,w_p)'\in\mathbb{R}^p,w_j=\text{Cor}(X_{.j},y)$ for variable screening
by selecting the variable set $\mathcal{A}_\gamma = \{j\in [p]:|w_j|>\gamma\}$ depending on a threshold $\gamma>0$, where $[p]=\{1,\dots,p\}$. Under certain 
technical conditions, where $p$ grows exponentially with $n$, they show that this procedure has the \textit{sure screening property} 
%\begin{align}
$\mathbb{P}(\mathcal{A} \subset \mathcal{A}_{\gamma_n})\to 1 \text{ for } n\to \infty$
%\end{align}
with an explicit exponential rate of convergence, where $\mathcal{A}=\{j\in[p]:\beta_j\neq 0\}$ is the set of truly active variables. These conditions imply that 
$\mathcal{A}$ and $\mathcal{A}_{\gamma_n}$
contain less than $n$ variables. One of the critical conditions is that on the population level for some fixed $i\in[n]$,
\begin{align}\label{eq:condSIS}
  \min_{j\in\mathcal{A}}|\text{Cov}(y_i/\beta_j,x_{ij})| \geq c
\end{align}
for some constant $c>0$, which rules out practically possible
scenarios where an important variable is marginally uncorrelated to the response. 

If we want a screening measure for marginal variable importance considering the other variables in the model, 
one natural choice in a usual linear regression model with $p<n$ would be the least-squares estimator $\hat\beta = (X'X)^{-1} X'y$.
The Ridge estimator $\hat \beta_\lambda = (X'X+\lambda I_p)^{-1} X'y$,
can be seen as a compromise between the two, since $\lim_{\lambda\to 0} \hat \beta_\lambda = \hat \beta$ and $ \lim_{\lambda\to\infty}\lambda \hat \beta_\lambda = X'y$.
It can also be used in the case $p>n$ and has the alternative form (see Lemma~\ref{lemma:altRidge})
\begin{align}
  \hat \beta_\lambda = X'(\lambda I_n+ XX')^{-1} y,
\end{align}
which is especially useful for saving computational complexity for very large $p$, since the inverted matrix only has dimension $n\times n$, bringing down the computational complexity to $\mathcal{O}(n^2p)$ \citep{Wang2015HOLP}.
If we now let $\lambda\to 0$, assuming $\text{rank}(XX')=n$ and therefore $p>n$, we end up with the HOLP estimator from \citet{Wang2015HOLP}
\begin{align}
  \hat \beta_{\text{HOLP}} = X'(XX')^{-1} y = \lim_{\lambda \to 0} \hat \beta_\lambda,
\end{align}
which is also the minimum norm solution to $X\beta = y$ (see Lemma~\ref{lemma:HOLPminNorm}).
\citet{Kobak2020RidgeinHD} show that the optimal Ridge penalty for minimal mean-squared prediction error can be zero or negative
for real-world high-dimensional data, because low-variance directions in the predictors can already provide an implicit Ridge regularization.

This motivates choosing the absolute values of the tuning-free coefficient vector $\hat\beta_{\text{HOLP}}$ for variable screening. 
Under similar conditions as in \citet{Fan2007SISforUHD}, but without assumption \eqref{eq:condSIS} on the marginal correlations to the response,
and allowing $p>c_0 n$ with $c_0>1$ to grow at any rate, \cite{Wang2015HOLP} show that $\hat\beta_{\text{HOLP}}$ also has the sure screening property.
Furthermore, they show \textit{screening consistency} of the estimator for exponential growth of $p$, meaning that 
\begin{align}
  \mathbb{P}\Big(\min_{j \in \mathcal{A}} |\hat\beta_{\text{HOLP},j}|>\max_{j \notin \mathcal{A}} |\hat\beta_{\text{HOLP},j}| \Big)\to 1 \text{ for } n\to \infty
\end{align}
at an exponential rate. This means that, asymptotically, the $a=|\mathcal{A}|$ highest absolute coefficients of $\hat\beta_{\text{HOLP}}$ correspond exactly to the true 
active variables. 

In another work, \citet{Wang2015ConsistencyHOLP} derive and compare the requirements for SIS and HOLP screening to have the \textit{strong screening consistency}
\begin{align}
  \min_{j \in \mathcal{A}} |\hat\beta_j|>\max_{j \notin \mathcal{A}} |\hat\beta_j| \quad\text{ and }\quad
  \text{sign}(\hat\beta_j) = \text{sign}(\beta_j) \quad \forall j \in \mathcal{A},
\end{align}
where $\hat\beta$ is an estimator of the true $\beta$. Both methods are shown to be strongly screening consistent with large probability when the sample size
is of order $n=\mathcal{O}\big((\rho a + \sigma/\tau)^2\log(p)\big)$, where $\tau= \min_{j\in\mathcal{A}} |\beta_j|$ measures the signal strength and 
$\rho = \max_{j\in\mathcal{A}} |\beta_j| / \tau$ measures the diversity of the signals. However, for SIS the predictor covariance matrix $\Sigma$ 
needs to satisfy the \textit{restricted diagonally dominant} (RDD) condition given in \citet[Definition 3.1]{Wang2015ConsistencyHOLP}, which is related to the
\textit{irrepresentable condition} (IC) of \citet{Zhao2006LASSOModSel} for model selection consistency of the LASSO.  
This excludes certain settings, while for HOLP only the  condition number $\kappa$ of the covariance matrix enters the required sample size as a constant, meaning 
it is always screening consistent for a large enough sample size.

In the calculation of $\hat\beta_{\text{HOLP}}$ there might be a problem when $p$ is close to $n$ or $XX'$ is close to degeneracy, which can lead to a blow-up of the error term.
In the discussion, \citet{Wang2015HOLP} recommend to use the Ridge coefficient $\hat \beta_\lambda = X'(\lambda I_n+ XX')^{-1} y$ with
penalty $\lambda=\sqrt{n} + \sqrt{p}$ in this case to control the explosion of the noise term.

So far, we have shown theoretical foundations for HOLP and SIS screening. 
Now we also want to look at the practical performance
in a quick simulation example. 
The simulation study provided in \cite{Wang2015HOLP} focuses on correctly selecting a sparse true model, while
we are also interested in the HOLP estimator being almost proportional to the true regression coefficients $\beta$ for later application in the random projection. 

Therefore, we simulate data similar to later in Section~\ref{sec:simulation:datagen} from the following example used throughout this section.
\begin{example}\label{ex:data_setting}
  We generate data from \eqref{eqn:datastr} with multivariate normal predictors $x_i\sim N(0,\Sigma)$ and normal errors $\varepsilon_i\sim N(0,\sigma^2)$, where we choose
  $n=200, p=2000, a=100,\mu=1$, and $\Sigma = \rho 1_p1_p' + (1-\rho)I_p$ has a compound symmetry structure with $\rho=0.5$ and eigenvalues $\lambda_1 = 1-\rho + p\rho, \lambda_j=1-\rho,j=2,\dots,p$. 
  The first $a=100$ entries of $\beta$ are uniformly drawn from $\pm\{1,2,3\}$ and the rest are zero. 
  The error variance $\sigma^2$ is chosen such that the signal-to-noise ratio is $\rho_{\text{snr}}=\beta'\Sigma\beta/\sigma^2=10$. 
\end{example}

We compare variable screening based on the marginal correlations, HOLP, Ridge with proposed penalty $\lambda =\sqrt{n} + \sqrt{p}$ and 
Ridge with $\lambda$ chosen by $10$-fold cross-validation. Figure~\ref{fig:ScreeningDens} shows density estimates of the absolute coefficients estimated by these four methods
for truly active and non-active variables for $100$ replicated draws of the data.
In Figure~\ref{fig:ScreeningMeasures} we evaluate the selection process of the four methods when selecting the $k$ variables
having the highest absolute estimated coefficients 
and let $k$  vary on the x-axis. 
We show the precision and recall of this selection, as well as the ratio of correct signs for truly 
active predictors included in the selection and the correlation of the corresponding true 
coefficients to the estimates averaged over the $100$ replications.
We see that HOLP and Ridge with penalty $\lambda=\sqrt{n}+\sqrt{p}$  better separate the active and non-active predictors and achieve better results for precision, recall,
true sign recovery and correlation to the true coefficient compared to Ridge with cross-validated penalty and correlation-based screening.
In Figure~\ref{fig:ScreeningDens}, we see that the absolute coefficients of cross-validated Ridge are much smaller than HOLP and Ridge with $\lambda=\sqrt{n}+\sqrt{p}$, meaning the $\lambda$ suggested by cross-validation is much higher.
In comparison, the choice $\lambda=\sqrt{n}+\sqrt{p}$ even leads to quite similar results as HOLP, which can be interpreted as Ridge with $\lambda=0$.

\begin{figure*}[t!]
  \begin{subfigure}{.495\textwidth}
    \centering
    \includegraphics[width=1\textwidth,trim=0 0 0 0, clip]{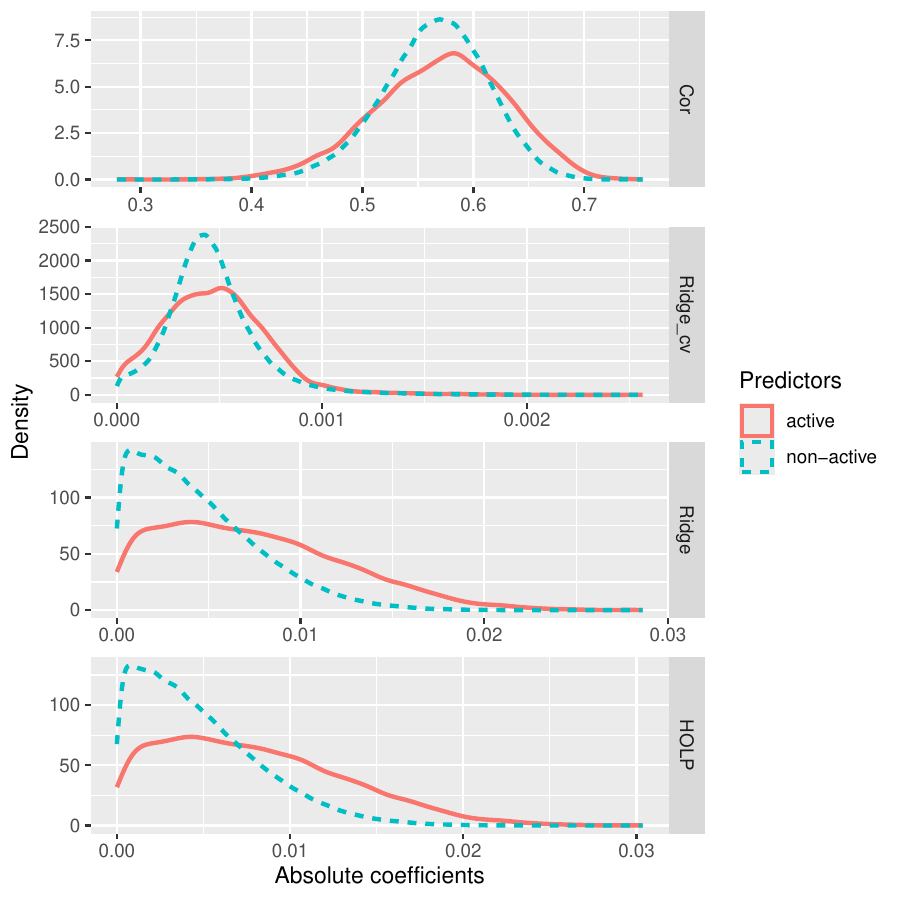}
    \caption{Density estimates of absolute coefficients}
    \label{fig:ScreeningDens}
  \end{subfigure}
  \begin{subfigure}{.495\textwidth}
    \centering
    \includegraphics[width=1\textwidth,trim=0 0 0 0, clip]{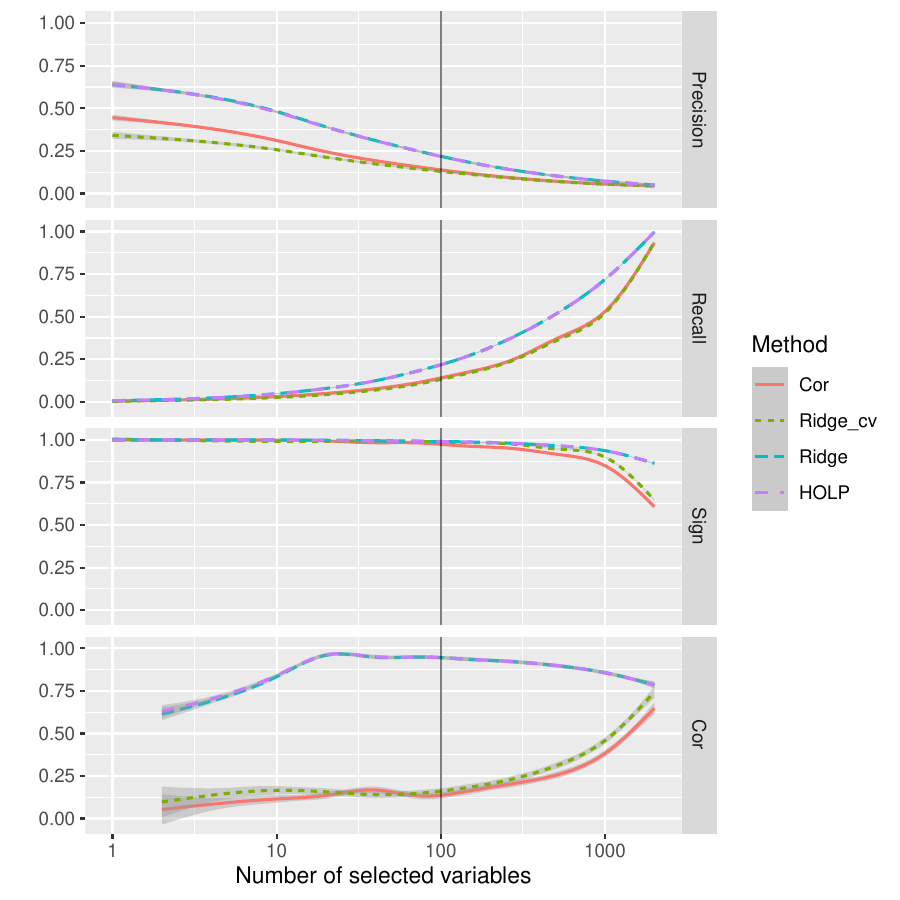}
    \caption{Screening measures}
    \label{fig:ScreeningMeasures}
  \end{subfigure}
  \caption{Comparison of screening based on marginal correlations, HOLP, Ridge with $\lambda =\sqrt{n} + \sqrt{p}$ and 
  Ridge with cross-validated $\lambda$ in the setting in Example~\ref{ex:data_setting}. (a) shows density estimates of absolute estimated coefficients for active and non-active predictors over $n_\text{rep}=100$ repetitions. (b) shows 
  precision, recall, sign recovery, and correlation of estimates to the true coefficients averaged over $100$ replications, where the vertical line indicates the true number of active variables.}
  \label{fig:Screening}
  \end{figure*}

\subsection{Random Projection}\label{sec:method:RP}

Random projection is used as a dimension reduction tool in high-dimensional statistics by creating a random matrix $\Phi\in\mathbb{R}^{m\times p}$ with $m\ll p$ and using the reduced predictors
$z_i=\Phi x_i \in\mathbb{R}^m$ for further analysis. When applying it to linear regression, we would wish that the reduced 
predictors still have most of the predictive power and that $\beta\in\text{span}(\Phi')$, such that the true coefficients can still be recovered after the reduction.

Random projections first became popular after \citet{JohnsonLindenstrauss1984}, hereafter abbreviated by \textit{JL}, who proved the existence of a linear map that approximately preserves pairwise distances for a 
set of points in high dimensions in much lower-dimensional space. Many papers followed, giving explicit constructions of a random matrix $\Phi\in\mathbb{R}^{m\times p}$
satisfying this \textit{JL} property with high probability. The classic construction is setting the elements of this matrix $\Phi_{ij} \overset{iid}{\sim} N(0,1)$ \citep{FRANKL1988JLSphere}, but also 
sparse versions with iid entries
\vspace{-3mm}
\begin{align}\label{eq:sparseRMgen}
  \Phi _{ij} = \begin{cases}
      \pm 1/\sqrt{\psi} & \text{with prob. } \psi/2 \\
      0 & \text{with prob. } 1 - \psi
    \end{cases},
  \end{align}
  for $0<\psi\leq 1$ can satisfy the property after appropriate scaling.
  \citet{ACHLIOPTAS2003JL} gave the results for $\psi=1$ or $\psi=1/3$, and even sparser choices such as $\psi = 1/\sqrt{p}$ or $\psi = \log(p)/p$ were later 
  shown to be possible with little loss in accuracy of preserving distances \citep{LiHastie2006VerySparseRP}.

In this section, we propose a new random projection matrix tailored to the regression problem. We start from a \textit{sparse embedding matrix} $\Phi\in\mathbb{R}^{m\times p},m\ll p$ from
\citet{Clarkson2013LowRankApprox}, which can be used to form a random projection with the \textit{JL} property and is obtained in the following way.
\begin{definition}\label{def:RP_CW}
  Let $h:[p]\to[m]$ be a random map such that for each $j\in[p]: h(j)=h_j \overset{iid}{\sim} \text{Unif}([m])$.
  Let $B\in\mathbb{R}^{m\times p}$ be a binary matrix with $B_{h_j,j}=1$ for all $j\in[p]$ and remaining entries $0$, where we assume $\text{rank}(B)=m$.
  Let $D\in \mathbb{R}^{p\times p}$ be a diagonal matrix with entries $d_j \sim \text{Unif}(\{-1,1\}), j\in[p]$ independent of $h$.
  Then we call $\Phi = BD$ a CW random projection.
\end{definition}
Each variable $j$ is mapped to a uniformly random goal dimension $h_j$ with random sign. 
We assume that each goal dimension $k\in[m]$ is attained by $h$ for some variable $j\in[p]$, which leads to $\text{rank}(B)=m$. Otherwise, we
just discard this dimension and reduce $m$ by one. When using this random projection
for our linear regression problem \eqref{eqn:datastr}, variables in the same goal dimension should not have signs conflicting their respective 
influence on the response, and, in general, we would wish for $\beta\in\text{span}(\Phi')$ such that the true coefficients $\beta\in\mathbb{R}^p$ can be 
recovered by the reduced predictors $z_i=\Phi x_i$ when modeling the responses as their linear combination
 $y_i \approx z_i'\gamma = x_i'\Phi'\gamma, \gamma\in\mathbb{R}^m$. 

 Lemma~\ref{lemma:ProjPhi} shows that for a CW random projection $\Phi$ with general
 diagonal entries $d_j\in\mathbb{R}$, the projection of a general $\beta\in\mathbb{R}^p$ to the row-span of $\Phi$ given by
 $\tilde \beta = P_\Phi\beta=\Phi'(\Phi\Phi')^{-1}\Phi \beta$ can be explicitly expressed as
 \vspace*{-3mm}
 \begin{align*}
  \tilde \beta_j = d_j \cdot \frac{\sum_{k:h_k=h_j}d_k\beta_k}{\sum_{k:h_k=h_j}d_k^2}.
 \end{align*}
 Therefore, we propose to set $d_j= c\cdot\beta_j$ for some constant $c\in\mathbb{R}$. We obtain 
 $\tilde \beta = \beta$ and therefore $\beta\in\text{span}(\Phi')$. The following theorem shows that we can improve the 
 mean square prediction error when using these diagonal elements proportional to $\beta$ instead of using random signs.

\begin{theorem}\label{theorem:1}
  Assume we have data $(y_i,x_i),i=1,\dots,n$ from the model
  \begin{align}
    y_i = x_i'\beta + \varepsilon_i, i=1,\dots,n,
  \end{align}
  where $x_i \overset{iid}{\sim} N(0,\Sigma)$ with $0<\Sigma\in\mathbb{R}^{p\times p},p>n$ and the $\varepsilon_i$s are iid error terms with $\mathbb{E}[\varepsilon_i]=0$ and constant $\text{Var}(\varepsilon_i)=\sigma^2$ independent of the $x_i$s,
  and we want to predict a new observation from the same distribution $\tilde y = \tilde x'\beta + \tilde \varepsilon$ independent from the given data.
  For a smaller dimension $m<n-1$, let $\Phi_{\text{rs}}=BD_{\text{rs}}\in\mathbb{R}^{m\times p}$ be the CW random projection with random sign diagonal entries and 
  $\Phi_{\text{pt}}=BD_{\text{pt}}\in\mathbb{R}^{m\times p}$ 
  the CW random projection with diagonal entries $d_j^{\text{pt}}=c\beta_j$ for some constant $c>0$ proportional to the true coefficient $\beta$. 
  We assume that for each $i\in[m]$ there is a $j\in h^{-1}(i)=\{k\in[p]:h(k)=i\}$ with $\beta_j\neq 0$. Otherwise, in order to retain $\text{rank}(\Phi_{\text{pt}})=m$, we set
  $j_i = \min(h^{-1}(i))$ and $d_{j_i}^{\text{pt}} = \text{Unif}(\{-1,1\})\cdot \min_{j:d_j^{\text{pt}}\neq0}|d_j^{\text{pt}}|$ for each $i\in[m]$ where it does not hold.
   
   For $X\in\mathbb{R}^{n\times p}$ with rows $\{x_i\}_{i=1}^n$ and $y=(y_1,\dots,y_n)'\in\mathbb{R}^n$, let $Z_{\text{rs}} = X\Phi_{\text{rs}}'\in\mathbb{R}^{n\times m}$ and $Z_{\text{pt}} = X\Phi_{\text{pt}}'\in\mathbb{R}^{n\times m}$ be the reduced predictor matrices and
    $\hat y_{\text{rs}} = (\Phi_{\text{rs}}\tilde x)'(Z_{\text{rs}}'Z_{\text{rs}})^{-1} Z_{\text{rs}}'y$ and $\hat y_{\text{pt}} = (\Phi_{\text{pt}}\tilde x)'(Z_{\text{pt}}'Z_{\text{pt}})^{-1} Z_{\text{pt}}'y$  the corresponding least-squares predictions.  
    Then,
    \begin{align}\label{eqn:TH1}
      \mathbb{E}[(\tilde y -& \hat y_{\text{rs}})^2] - \mathbb{E}[(\tilde y - \hat y_{\text{pt}})^2] \geq C_\text{Th1}  >0, \\
      C_\text{Th1} &= \|\beta\|^2 \Bigl[\lambda_p(1-\frac{2m}{p})\Bigr] +
       \frac{a}{p-1}m\lambda_p \tau^2  (1-\frac{m+1}{p-1}+ \mathcal{O}(p^{-2})),
    \end{align}
    where $\mathcal{A}=\{j\in[p]:\beta_j\neq 0\}$ is the active index set, $a=|\mathcal{A}|$ is the number of active variables, $\tau=\min_{j:\beta_j\neq 0}|\beta_j|$ is the smallest non-zero absolute coefficient and $\lambda_p>0$ is the smallest eigenvalue of $\Sigma$.
\end{theorem}
The proof can be found in Appendix~\ref{sec:AppA}. 
% \peter{It would be good to explain a bit in words what you have shown with the proof, and what are the implications. This could go into the remarks below.}
\begin{remark}\label{rem:Th1}

  \begin{itemize}
    \item This theorem shows that when using a conventional random projection from Definition~\ref{def:RP_CW} for least-squares regression, 
    the expected squared prediction error is much smaller when using diagonal elements proportional to the variables' true effect to the response as opposed to the conventional random sign,
    and gives an explicit conservative lower bound on how much smaller it has to be at least.
    \item In practice, the true $\beta$ is unknown, but in Section~\ref{sec:method:Scr} we saw that $\hat\beta_{\text{HOLP}}$ asymptotically
    recovers the true sign and order of magnitude with high probability, and has high correlation to the true $\beta$, meaning it is `almost' proportional to the true $\beta$. So we propose to use $\hat\beta_\text{HOLP}$ as diagonal elements of our projection.
    See Remark~\ref{rem:ProofTh1} in Appendix~\ref{sec:AppA} for a short note on the implications on the error bound, the relaxation of distributional assumptions and the full-rank adaption of $\Phi_{\text{pt}}$.
    \item Note that this bound is non-asymptotic and valid for any allowed $m,n,p,a$ (up to the quadratic order in $p$),
    and it does not depend on the signal-to-noise ratio $\rho_\text{snr}$ or the noise level $\sigma^2$, because they have the same average effect on the error 
    for both random projections.
  \end{itemize}  
\end{remark}

In the following, we want to verify above considerations and the obtained bound by evaluating the prediction performance of different projections in a small simulation example, where we use the setting 
from Example~\ref{ex:data_setting} again.
When $\Phi\in\mathbb{R}^{m \times p}$ is the selected random projection matrix, we fit an ordinary least-squares
model to the responses $y_i$ on the reduced predictors $z_i=\Phi x_i$ to obtain predictions for $n_{test}=100$ new predictor observations.
These predictions are evaluated by the mean squared prediction error MSPE.
We set the reduced dimension to the true number of active variables $m=a=100$ and 
compare $\Phi$ chosen Gaussian with iid $N(0,1)$ entries, Sparse from \eqref{eq:sparseRMgen} with $\psi=1/3$, and the following three versions
from our Definition~\ref{def:RP_CW}: SparseCW with standard random sign diagonal elements, SparseCWSignH with $d_j=\text{sign}(\hat\beta_{\text{HOLP},j})$ and SparseCWHolp with $d_j=\hat\beta_{\text{HOLP},j}$. 
Additionally, we look at 
%regression with HOLP, which uses the full predictors, and 
two oracles SparseCWSignB from Definition~\ref{def:RP_CW} with $d_j=\text{sign}(\beta_j)$ and 
SparseCWBeta with $d_j=\beta_j$ with the full-rank adaptions proposed in Theorem~\ref{theorem:1}.
Figure~\ref{fig:RPrMSPE} shows prediction performance of these different projections for $100$ replications.
We also plot the theoretical lower bound $C_\text{Th1}$ from Theorem~\ref{theorem:1} from the best oracle to SparseCW with random signs and see that the difference is actually higher.
The conventional random projections stay well above this bound, while our proposed random projections using the HOLP-coefficient 
manage to stay within the bound to the oracle's performance,
% SparseCWHolp is able to produce predictions that are as good as the ones 
% from the HOLP model using all variables, 
with random projection using the sign-information (SparseCWSignH) instead of the coefficients performing only slightly worse than SparseCWHolp.

% \begin{wrapfigure}{r}{0.5\textwidth}
  \begin{figure}[t!]
  \centering
  \includegraphics[width=0.8\textwidth,trim=0 20 0 20, clip]{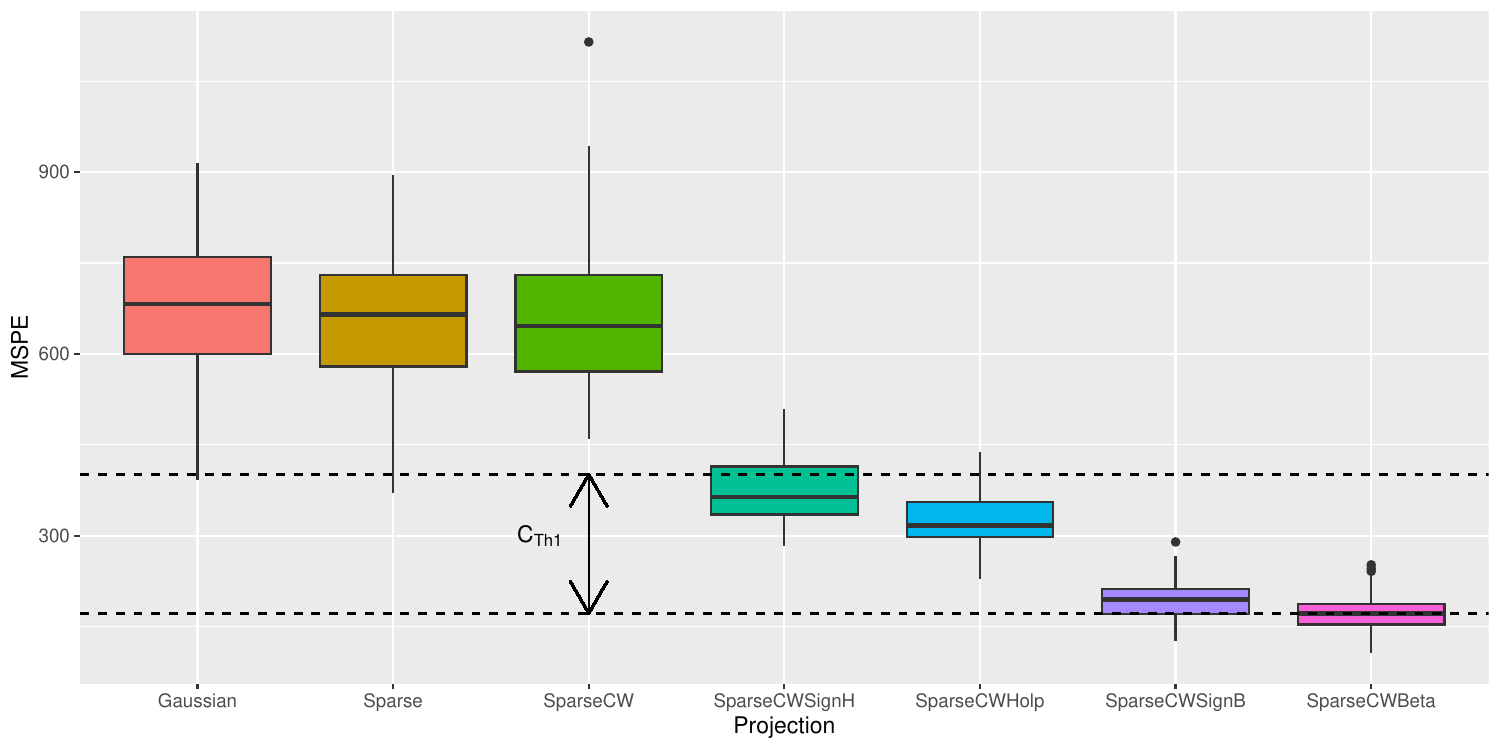}
  \caption{Mean squared prediction errors of different conventional projections, the proposed projection using the HOLP coefficient (SparseCWHolp) or its signs (SparseCWSignH), and the oracle projections using the true $\beta$ (SparseCWBeta) or its signs (SparseCWSignB) for $n_\text{rep}=100$, $n=200, p=2000, m=a=100$ and $\Sigma=0.5\cdot 1_p1_p' + 0.5 \cdot I_p$.}
  \label{fig:RPrMSPE}
\end{figure}
% \end{wrapfigure}

\subsection{Combination of Screening and Random Projection}\label{sec:method:ScrRp}

% Previous work by \citet{Dunson2020TargRandProj} in this area showed that it can be beneficial to combine these two tools for 
% dimension reduction by using a probabilistic variable screening step first, keeping only the more important ones, and then performing random projection 
% on these remaining variables. Repeating these steps many times and averaging results can reduce dependence on the random projection and increase prediction performance.
% After explaining their methods in more detail, we will propose our adaptions of the procedure and investigate different combinations of screening and random projection and the effect of the 
%probabilistic variable screening step first, keeping only the more important ones, and then performing random projection 
%on these remaining variables. Repeating these steps many times and averaging results can reduce dependence on the random projection and increase prediction performance number of screened variables.

 \citet{Dunson2020TargRandProj} propose a two-step approach, where a probabilistic variable screening step is performed first. 
In more detail, they propose to include each variable with probability 
\begin{align*}
  q_j = \frac{|\text{Cor}(x_{ij},y_i)|^\nu}{\max_k |\text{Cor}(x_{ik},y_i)|^\nu}, \quad \nu>0, \text{ some } i\in[n]
\end{align*}
with $\nu= (1+\log(p/n))/2$. In a second step they use a random dimension \linebreak $m \sim \text{Unif}(\{2\log(p),\dots,3n/4\})$ for a general sparse projection matrix of type \eqref{eq:sparseRMgen}. 
With this choice, the variable with highest marginal importance is always included and the number of screened variables is not
directly controlled. They repeat these steps many times and average the results. There is no explicit discussion on the number of models used in the paper, but \citet{BayComprRegr} report that the gains 
are diminishing after using around $100$ models for averaging.

Instead, we propose to set the number of screened variables to a fixed multiple of the 
sample size $c \cdot n$ (independent of $p$), and drawing the variables with probabilities proportional to their utility based on the HOLP-estimator $p_j\propto |\hat\beta_{\text{HOLP},j}|$,
as well as using slightly smaller goal dimensions $m \sim \text{Unif}(\{\log(p),\dots,n/2\})$ to increase estimation performance of the linear regression in the reduced model, 
and our proposed random projection from Definition~\ref{def:RP_CW} with the entries of $\hat\beta_\text{HOLP}$ corresponding to the screened variables as diagonal elements. These steps are explained more rigorously in Section~\ref{sec:method:SPAR}.

We go back to our data setting from Example~\ref{ex:data_setting} and want to examine the effects of the number of marginal models for
different combinations of variable screening and random projection for our proposed adaptations. Figure~\ref{fig:ScRPrMSPE} shows the effect of 
the number of models used on the average prediction performance over $100$ replications and compares the following four methods: 
screening to $n/2$ variables based on $\hat\beta_\text{HOLP}$ (Scr\_HOLP), random projections with SparseCW matrix (RP\_CW), and
first screening with $\hat\beta_\text{HOLP}$ to $2n$ variables and then using the conventional SparseCW random projection (ScrRP\_CW) or our proposed SparseCWHolp random projection (ScrRP). 
When we use just one model, the screening methods deterministically select the variables with highest marginal importance $|\hat\beta_{\text{HOLP},j}|,j=1,\dots,p$, 
otherwise they are drawn with probabilities $p_j\propto |\hat\beta_{\text{HOLP},j}|$, as previously mentioned. We can see that the combination of 
screening and the proposed random projection yields the best performance and the effect of using more models diminishes at around $20$ models already for this method.

In Figure~\ref{fig:ScRPnscMSPE} we look at the effect of the number of screened variables $c\cdot n$ on prediction performance, where
we compare just screening (Scr\_HOLP) to the combination of screening with random projection (ScrRP) as above for a fixed number of models $M=20$, and we show again averages over $100$ replications.
For just screening we use the HOLP estimator from Section~\ref{sec:method:Scr} as the subsequent regression method when $c\geq 1$, and in 
case the system is close to degeneracy we add a small ridge penalty $\lambda=0.01$ to the OLS estimate. We see that the screening still has bad 
performance for $c$ close to $1$, because the sample covariance of the selected predictors is close to singularity. For small and large ratios it achieves better prediction performance. 
When combining the screening with the random projection, $c$ does not have such a large impact, and we can achieve lower prediction errors
where the best results are achieved for $2 \leq c \leq 4$.

So far, every variable selected once in the screening step will have a contribution in the final regression coefficient, so when we choose a smaller number of used models and ratio $c$, there will be fewer variables employed in the model.
To achieve additional sparsity, we will 
use a thresholding step to actively set less important contributions to $0$ to obtain a higher level of sparsity (see Section~\ref{sec:method:SPAR}).

When combining the predictions of different models, there are many different ways to choose their respective model weights, such as
AIC \citep{BurnhamAnderson2004MMIAIC}, prediction error (leave-out-one or cross-validation), true posterior model weights in a Bayesian approach or 
dynamic model weights in time series modeling \citep{gruber2023forecasting}.
Alternatively, a measure of diversity of the models could be taken into account \citep[e.g.,][]{Reeve2018DivAndDFinRegEnsembles}. 
However, across all our efforts, the simple average across all models turned out to yield the best predictions for the investigated settings. 
Similar observations were already reported in the literature as the forecast combination puzzle \citep{CLAESKENS2016ForecastCombPuzzle}.

  \begin{figure}[t!]
  \centering
  \includegraphics[width=0.8\textwidth,trim=0 5 0 20, clip]{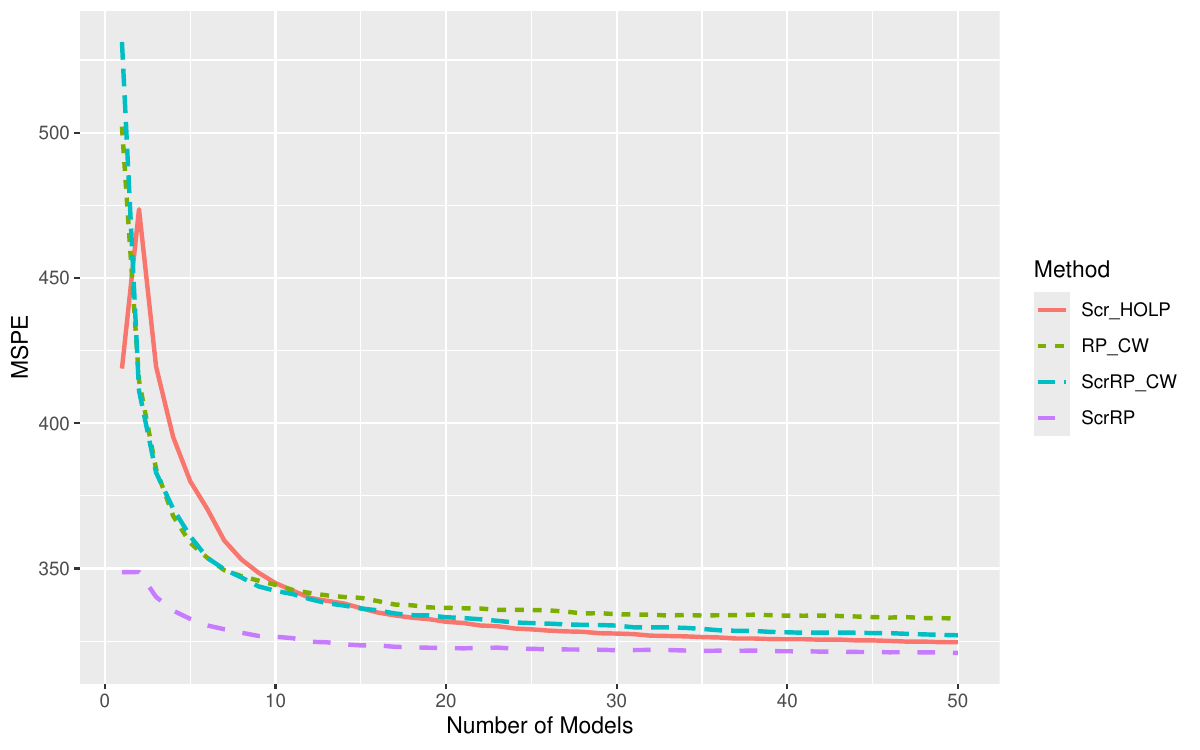}
  \caption{Average mean squared prediction error for using ensembles of screening using HOLP (Src\_HOLP), CW random projections (RP\_CW), screening with HOLP and conventional random projections (ScrRP\_CW) and the combination of screening using HOLP and our proposed CW random projection (ScrRP) for different number of models for $n_\text{rep}=100$, $n=200, p=2000, m=a=100$ and $\Sigma=0.5\cdot 1_p1_p' + 0.5 \cdot I_p$.}
  \label{fig:ScRPrMSPE}
\end{figure}

  \begin{figure}[t!]
  \centering
  \includegraphics[width=0.8\textwidth,trim=0 0 0 0, clip]{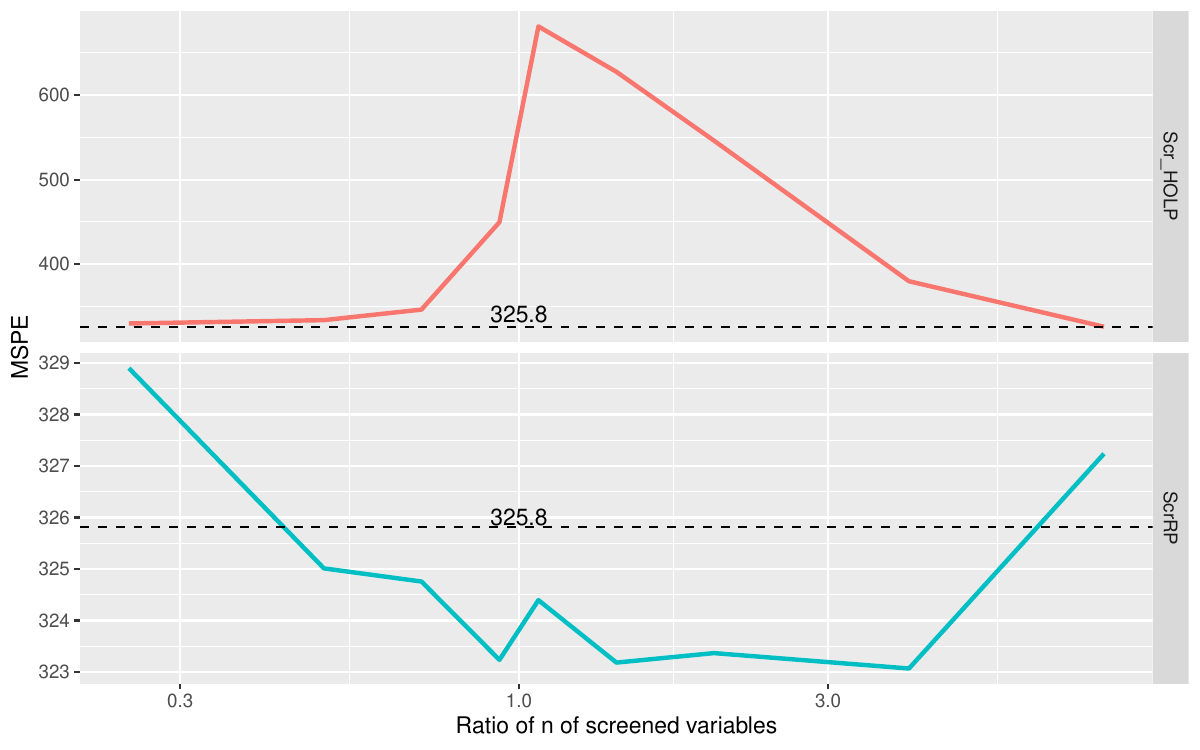}
  \caption{Average effect of number of screened variables on mean squared prediction error of only screening compared 
  to screening plus random projection before fitting linear regression model for $100$ replications of
  the setting of Example~\ref{ex:data_setting}.
  %with $n=200, p=2000, m=r=100$ and $\Sigma=0.5\cdot 1_p1_p' + 0.5 \cdot I_p$
  }
  \label{fig:ScRPnscMSPE}
\end{figure}

\subsection{Sparse Projected Averaged Regression (SPAR)} \label{sec:method:SPAR}

The considerations of the previous sections lead us to propose the following algorithm for 
high-dimensional regression where $p>n$.
  \begin{enumerate}
    \item standardize inputs $X:n\times p$ and $y: n\times 1$ 
    \item calculate $\hat\beta_{\text{HOLP}} = X'(XX')^{-1} y$
    \item For $k=1,\dots,M$:
     \begin{enumerate}
      \item[3.1.] draw $2n$ predictors with probabilities $p_j\propto |\hat\beta_{\text{HOLP},j}|$ yielding screening index set $I_k=\{j_1^k,\dots,j_{2n}^k\}\subset[p]$
      \item[3.2.] project remaining variables to dimension $m_k\sim \text{Unif}\{\log(p),\dots,n/2\}$ using $\Phi_k:m_k\times 2n$ from Definition~\ref{def:RP_CW} with diagonal elements 
      $d_i=\hat \beta_{\text{HOLP},j_i^k}$
      to obtain reduced predictors $Z_k=X_{.I_k}\Phi_k' \in \mathbb{R}^{n\times m_k}$
      \item[3.3.] fit OLS of $y$ against $Z_k$ to obtain $\gamma^k=(Z_k'Z_k)^{-1} Z_k'y$ and $\hat \beta^k$, where $\hat \beta_{I_k}^k=\Phi_k'\gamma^k$ and $\hat \beta_{\bar I_k}^k=0$.
    \end{enumerate}
    \item for a given threshold $\lambda>0$, set all entries $\hat\beta_j^k$ with $|\hat\beta_j^k|<\lambda$ to $0$ for all $j,k$
    \item combine via simple average $\hat \beta = \sum_{k=1}^M\hat \beta^k / M$ 
    \item choose $M$ and $\lambda$ via $10$-fold cross-validation by repeating steps 1 to 5
    (but with using the original index sets $I_k$ and projections $\Phi_k$) for each fold and evaluating the prediction performance by MSE on the withheld fold; 
    and choose  
    \begin{align}
      (M_{\text{best}},\lambda_{\text{best}}) = \text{argmin}_{M,\lambda}\widehat{\text{MSE}}(M,\lambda)
    \end{align}

    \item output the estimated coefficients and predictions for the chosen $M$ and $\lambda$
  \end{enumerate}

  We use the following notation in step 3.
For a vector $y\in\mathbb{R}^n$ and an index set $I\subset [n]$, $\bar I$ denotes the complement, $y_I\in\mathbb{R}^{|I|}$ denotes the subvector with entries $\{y_i:i\in I\}$,
and for a matrix $B\in\mathbb{R}^{n\times m}$, $B_{I.}\in\mathbb{R}^{|I|\times m}$ denotes the submatrix with rows $\{B_{i.}:i\in I\}$ and similarly for a subset of the columns.

The standardization in step 1 helps to stabilize computation and makes the estimated regression coefficients comparable. 
% We can select the $\lambda$ which achieves the smallest MSE, or the
% $\lambda$ which leads to the least estimated active predictors, but still has MSE within one standard error of the best MSE.
The number of marginal models $M$ can also be chosen via cross-validation (after specifying a grid of values). However, in Figure~\ref{fig:ScRPrMSPE} we observed that the effect of $M$ diminishes after a certain value so it would be possible to 
fix it in the analysis. Note that higher values of $M$ will lead to more variables being employed in the ensemble. 
The thresholding step~4 introduces additional sparsity to the models in the ensemble where the threshold-level
can be selected via cross-validation. 
% to e.g., $M=20$.
% once it is high enough, and it 
% suffices to set $M=20$.

% \clearpage

\section{Simulation Study}\label{sec:simulation}

% This section compares different aspects of our proposed SPAR method to several competitors in an extensive simulation study.
%First, we explain the data generation setting including six different covariance structures and coefficient settings, ranging from sparse (few truly active variables) to dense (many active variables). Then, we introduce the evaluation measures and the considered competitors, before presenting the results in Section~\ref{sec:simulation:res}.

\subsection{Data Generation}\label{sec:simulation:datagen}

We assume the linear model in Equation~\eqref{eqn:datastr}.
% \begin{align}
%  y_i = \mu + x_i'\beta + \varepsilon_i,\quad i=1,\dots,n,
% \end{align}
% where $\mu$ is a deterministic constant, the $x_i\sim N_p(0,\Sigma)$ follow an iid $p$-variate normal distribution, and $\varepsilon_i \sim N(0,\sigma^2)$ are iid error terms 
% independent from the $x_i$s. 
The covariance matrix $\Sigma$ of the predictors and the coefficient vector $\beta\in\mathbb{R}^{p\times p}$ 
will change depending on the simulation setting. The intercept is set
to $\mu = 1$ and the error variance $\sigma^2$ is 
chosen such that the signal-to-noise ratio $\rho_{snr}=\beta'\Sigma\beta/\sigma^2=10$.
%\begin{align*}
%  \rho_{snr} = \frac{\text{Var}(\mu+x'\beta)}{\text{Var}(\varepsilon)} = \frac{\beta'\Sigma\beta}{\sigma^2}
%\end{align*}
%is equal to $10$. 
We choose $p=2000$ as a high number of variables and consider the following different simulation settings for $\Sigma$. 
The choice of the number of truly active variables $a$ and their coefficients $\beta$ will be explained below.

\begin{enumerate}
  \item \textit{Independent predictors}: $\Sigma= I_p$.
  \item \textit{Compound symmetry structure}: $\Sigma = \rho 1_p1_p' + (1-\rho) I_p$, 
  where we set $\rho=0.5$.
  \item \textit{Autoregressive structure}: The $(i,j)$-th entry is given by $\Sigma_{ij} = \rho^{|i-j|}$  and we choose $\rho=0.9$. 
  This structure is appropriate if there is a natural order among the predictors and 
  two predictors with larger distance are less correlated, e.g., when they give measurements over time.
  \item \textit{Group structure}: Similarly to scheme II in \citet{Dunson2020TargRandProj}, $\Sigma$ follows a block-diagonal
  structure with blocks of $100$ predictors each, where the first half of the blocks has the compound structure from 
  setting 2 and the second half has the AR structure from setting 3. Only the very last block has identity structure corresponding to independent predictors within that block, and the predictors between 
  different blocks are independent.
  \item \textit{Factor model}: Inspired by model 4.1.4. in \citet{Wang2015HOLP}, we first generate a $p\times k$ factor matrix $F$ with $k=a$ and iid standard normal entries,
  and then set $\Sigma = FF' + 0.01\cdot I_p$. Here, dimension reduction of the predictors will be useful, because most 
  of the information lies within the $k$-dimensional subspace defined by $F$.
  \item \textit{Extreme correlation}: This setting is designed such that methods relying on marginal correlations have difficulty in finding any true active predictor.
  Similarly to example 4 in \citet{Wang2009_extrCor_SuperMData}, we create 
  each predictor variable $x_i$ the following way. 
  For $i=1,\dots,n$, let $z_{ij}\sim N(0,1)$ be iid standard normal 
  variables for $j=1,\dots,p$ and $w_{ij}\sim N(0,1)$ iid standard normal variables for $j=1,\dots,a$ 
  independent of the $z_{ij}$s.
  We then set
  \begin{align*}
    x_{ij} = \begin{cases}
      (z_{ij}+w_{ij})/\sqrt{2} & j\leq a \\
      (z_{ij} + \sum_{k=1}^a z_{ik})/\sqrt{a+1} & j > a
    \end{cases}.
  \end{align*}
  The marginal correlation of any active predictor $x_j,j\leq a$ to the response 
  is way smaller than that of any unimportant predictor $x_k,k>a$. The exact ratio between them is 
   $(j/a)\cdot 2^{-3/2}\cdot (a+1)^{-1/2} < 1$ for $j=1,\dots,a$.
\end{enumerate}
We vary $a$ between a \textit{sparse} $a=2\log(p)$, \textit{medium} $a=n/2+2\log(p)$ and \textit{dense} $a=p/4$ choice (rounded to closest integer). For settings 1 to 5,
the positions of the non-zero entries in $\beta$ are chosen uniform random (without replacement) in $[p]$
and these entries are independently set as $(-1)^u (4\log(n)/\sqrt{n} +|z|)$, where $u$ is drawn from a Bernoulli
distribution with probability of success parameter $p=0.4$ and $z$ is a standard normal variable. 
This choice was taken from \citet{Fan2007SISforUHD}, such that the coefficients are bounded away from $0$ and vary in sign and magnitude.
In setting 6, we choose the first $a$ predictors to be active with $\beta_j=j$ for $j=1,\dots,a$ and $\beta_k=0$ for $k>a$.

For each setting we generate $n=200$ observations and evaluate the performance on $n_\text{test}=100$ further 
test observations.  For setting 4, we also consider $p=500,10000$, $n=100,400$ as well as $\rho_{\text{snr}}=1,5$, and each setting is repeated $n_\text{rep}=100$ times.
% In the following section we will introduce the used error measures.

\subsection{Error Measures}\label{sec:simulation:measures}

We evaluate prediction performance on $n_{test}=100$ independent observations via \textit{relative mean squared prediction error}
  \begin{align}
    \text{rMSPE} = \sum_{i=1}^{n_{test}} (\hat y_i^{test} - y_i^{test})^2 \Big/ \sum_{i=1}^{n_{test}} (y_i^{test} - \bar y)^2,
  \end{align}
which is also used and motivated in \cite{SilinFan2022HDlinReg}.
% This measure scales the mean squared error by the error of a naive estimator $\hat \beta=0$, , in the sense that it is close to zero for growing sample size and dimension. 
% Therefore,
This measure gives an interpretable performance measure relative to the naive estimator $\hat \beta=0$, which has been shown to achieve a small mean squared error in some high-dimensional settings, and we want to achieve $\text{rMSPE}<1$ as small as possible.

To evaluate how well the methods are able to rank the variables we employ the absolute value of the estimated coefficient vector to compute the partial area under the receiver operating characteristic curve (pAUC) \citep[similar to][]{Wang2019HDRinPracticeStudy}. In the computation of pAUC we limit the number of false positives to $n/2$, which also allows for a fairer comparison between sparse and dense methods than AUC. In all presentations, we rescale pAUC to the interval [0,1] for better comparison.
% Moreover, in comparison the AUC,  the advantage of only looking at a constrained range of the curve is that for the sparse methods we need not rank too many variables with $\hat\beta=0$.

Finally, we also evaluate the sparse methods and methods that perform screening in terms of variable selection using precision (proportion of truly active predictors among estimated active predictors) and recall (proportion of correctly identified active predictors among truly active predictors). 
% High precision and high recall are both worth striving for, however, there are rarely methods that excel in both. \laura{TODO: citation?}

% For the evaluation of variable selection, which is a hard task in this high-dimensional setting, we use precision, recall, and F1 score. High precision and high recall are both worth striving for, however, there are rarely methods that excel in both. 
% The F1 score combines these two into one measure to evaluate a method's ability to identify true active variables.
% Furthermore, we also report the estimated number of active predictors for the different methods.

\subsection{Competitors}\label{sec:simulation:comp}

We compare the following set of methods:
\begin{enumerate}
  %\item \rom{ no HOLP \citep{Wang2015HOLP}, instead show PLS, (PCR similar, Ridge worse in performance)}
  \item AdLASSO using 10-fold CV \citep{Zou2006AdLASSO}, 
  \item Elastic Net with $\alpha=3/4$ using 10-fold CV \citep{Zou2005ElasticNet}, 
  \item SIS \cite[screening method]{Fan2007SISforUHD},
  \item Projected linear regression using one draw of a Sparse CW random projection matrix with dimension randomly drawn as in Step 3.3. of the SPAR algorithm,
  \item An ensemble of $M=100$ models of the projected linear regression in 4.,
  \item TARP \cite[targeted random projection method]{Dunson2020TargRandProj},
  \item SPAR with fixed $\lambda=0,M=20$,
  \item SPAR CV with $M\leq 100$ and $\lambda$ both chosen by cross-validation
  \item PLS 
\end{enumerate}
We also performed PCR and a linear model with Ridge penalty chosen by 10-fold CV, but we omitted them from the results for a more compact overview. PCR performed similarly to PLS in prediction, while Ridge performed worse than PLS in most settings. Moreover, we also replace the CW matrix in 4. by Gaussian and sparse conventional random projection but do not report the results as their performance was very similar to that of CW in all settings. 
% HOLP is a method for a non-sparse, i.e. dense, setting, because all estimated regression coefficients are non-zero, and does, therefore, not perform any variable selection. 
% The TARP method, in the way it is provided, does not return estimated regression coefficients, 
% but in principle each variable that is selected at least once in the screening will have non-zero coefficient and the method is therefore not suitable for variable selection as well.
Methods 1 to 3 can be considered sparse methods and will be marked by dotted boxes in the following figures. 

All methods were implemented in R \citep{RLanguage} using the packages \texttt{glmnet} \cite[AdLASSO and ElNet]{glmnetR}, 
\texttt{SIS} \citep{SISR}, \texttt{pls} \citep{plsR},
and the source code available online on \url{https://github.com/david-dunson/TARP} for TARP.
Our proposed method is implemented in the R-package \texttt{SPAR} available on GitHub (\url{https://github.com/RomanParzer/SPAR}).

\subsection{Results}\label{sec:simulation:res}
First, we look at the prediction results of the competing methods for the six different covariance settings and 
sparse, medium and dense active predictor settings with fixed $n=200,p=2000,\rho_\text{snr}=10$ in Figure~\ref{fig:rMSPE_cov_settings}. 
We see that the overall performance depends heavily on the covariance setting, and the signal-to-noise ratio alone does not 
quantify the difficulty of a regression problem. In the `independent' covariance setting with many active predictors, all methods barely outperform the naive estimator $\hat\beta=0$ with an rMSPE close to one, while in other covariance settings, the
errors are much lower. In general, we see that the sparse methods, especially AdLASSO and ElNet, perform well in sparse settings, but not in settings with more active variables.
On the other hand, the PLS method performs well in all dense settings, but less so in sparse settings. Except in some sparse settings, the SPAR method provides competitive results. 
To assess the overall performance, for each scenario and each repetition, we rank the methods from best (=1) to worst (=9) in terms of their relative MSPE. Table~\ref{tab:ranks} shows the average of these ranks (and its standard error). The proposed SPAR CV method has the best average rank followed by SPAR and PLS. It provides a good prediction performance all-around, proving that it is a viable option, especially in cases where it is not clear how sparse the problem is in practice.  We also observe that in terms of prediction ability, SPAR's performance is close to SPAR CV's. Depending on the application context, the additional computational cost of the cross-validation can be avoided with minimal loss in prediction power.
\begin{figure*}
  \centering
    \includegraphics[width=0.8\textwidth,trim=0 20 0 5, clip]{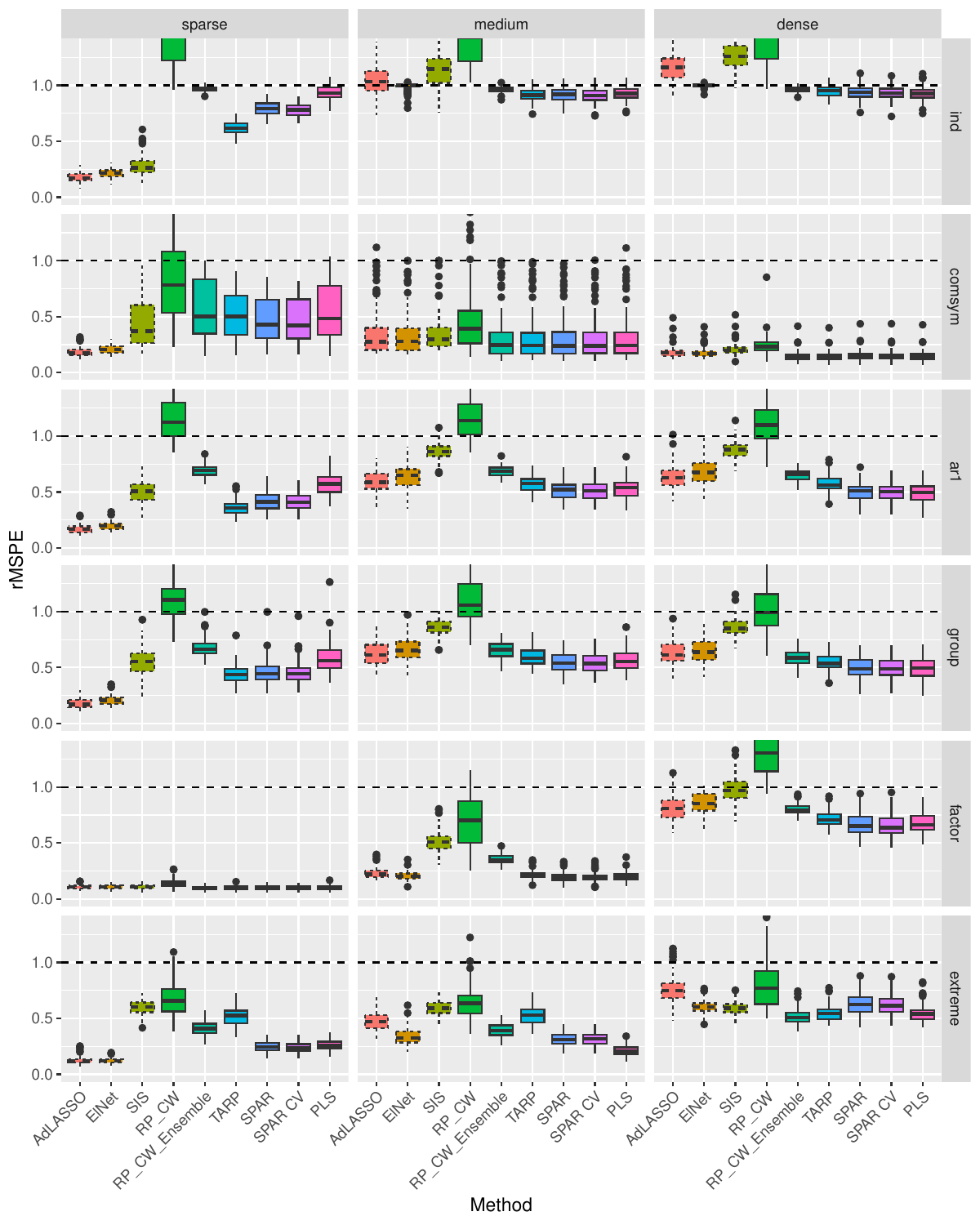}
  \caption{
  Relative MSPE of competing methods for different covariance and 
  active predictor settings for $n_\text{rep}=100$ replications ($n=200,p=2000,\rho_\text{snr}=10$). Sparse methods are marked by dotted boxes.}
  \label{fig:rMSPE_cov_settings}
\end{figure*}

\begin{table}[t!]
    \caption{Mean and standard error of the rank (best to worst) achieved by each method in terms of 
    rMSPE and pAUC across all investigated settings for fixed $n=200,p=2000,\rho_\text{snr}=10$ and 
    $n_\text{rep}=100$ replications. 
    The cells of the top 3 ranked methods are highlighted.}
    \label{tab:ranks} 
    \centering
\begin{tabular}{lHlHl}
\toprule
Method & \multicolumn{2}{c}{rMSPE} & \multicolumn{2}{c}{pAUC}\\% &Precision & Recall\\
\midrule
AdLASSO & 4.658 &(0.062) & 3.399 &(0.043) \\% &2.083 (0.02) & 3.256 (0.014)\\
ElNet & 4.577& (0.052) & 4.139 &(0.061)\\%  & 2.439 (0.025) & 3.525 (0.02)\\
SIS & 7.075 &(0.044) & 7.366 &(0.03) \\% & 2.262 (0.037) & 4.838 (0.008)\\
RP\_CW & 8.711 &(0.02) & 8.417 &(0.029) \\% & NA & NA\\
RP\_CW\_Ensemble & 5.516 &(0.048) & 6.085 &(0.042) \\% & NA & NA\\
%\addlinespace
TARP & 4.105& (0.038) & 5.066 &(0.052)\\%  & NA & NA\\
SPAR & 3.480 &(0.038) & 3.832& (0.04) \\% & 3.775 (0.017) & 1.96 (0.009)\\
SPAR CV & 3.143 & (0.038) & 3.141 &(0.038) \\% & 4.442 (0.02) & 1.42 (0.014)\\
PLS & 3.735 &(0.053) & 3.555 &(0.049) \\% & NA & NA\\
\bottomrule
\end{tabular}
\end{table}

Figure~\ref{fig:pAUC_cov_settings} provides information on how well the methods rank the variables as measured by pAUC. We observe that the results again highly depend on the investigated setting. The sparse methods achieve a good pAUC in most sparse settings, while SPAR CV performs well in almost all other settings. In Table~\ref{tab:ranks} we see that SPAR CV followed by AdLASSO and PLS achieve the best performances when ranking the methods from best to worst based on pAUC. 

\begin{figure*}
  \centering
    \includegraphics[width=0.8\textwidth,trim=0 20 0 5, clip]{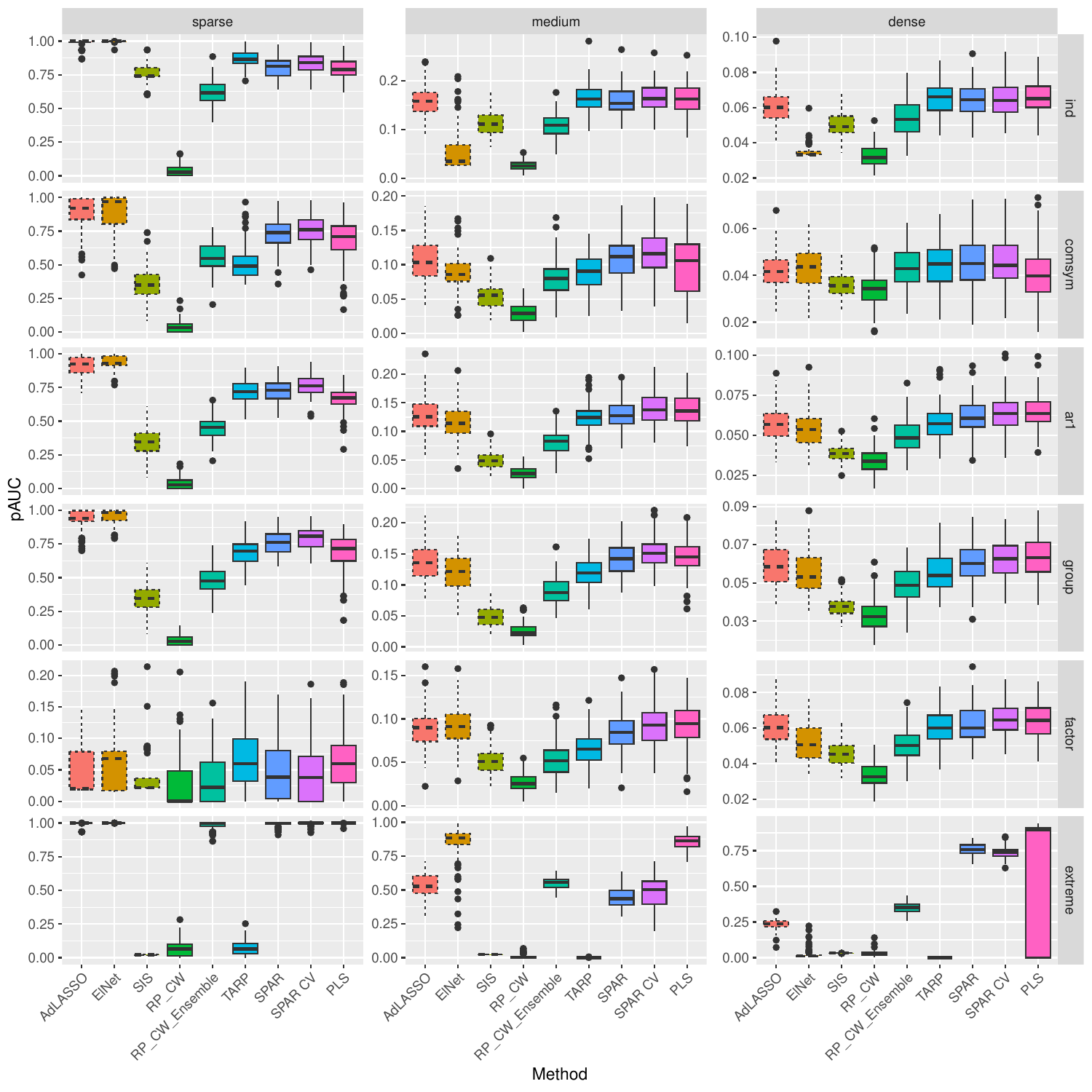}
  \caption{
  Partial AUC of competing methods for different covariance and 
  active predictor settings for $n_\text{rep}=100$ replications ($n=200,p=2000,\rho_\text{snr}=10$). Sparse methods are marked by dotted boxes.}
  \label{fig:pAUC_cov_settings}
\end{figure*}

In the Appendix we present the results of the sparse methods and the methods 
including a screening step in terms of variable selection by looking at precision (Figure~\ref{fig:Prec_cov_settings}) and recall
(Figure~\ref{fig:Rec_cov_settings}). We observe that SPAR achieves a high recall but lower precision, showing that the number of employed variables is rather high compared to the truly active predictors. The same can be observed for TARP. The sparse methods on the other hand achieve better precision than a dense method which would select all variables, and also a high recall in most sparse settings, but not in more dense settings.

% \laura{Add sme more on precision, recall and AUC?}
% Figure~\rom{not anymore, update text} shows precision, recall and F1 score of all competitors (except TARP, see remark in Section~\ref{sec:simulation:comp}) for
% the same covariance structures and medium setting for the active variables.  
% The sparse methods do achieve higher precision, while the
% dense methods reach higher recall. However, no method achieves a good combined F1 score in most settings, which suggests that variable selection in high dimensions is a 
% very hard task. Interestingly, the 'extreme' covariance setting is the only setting where some methods achieve good F1 scores.
% This setting was designed to make it hard for methods using marginal correlations of predictors to the response, and we can see
% that SIS does not select any true active variables.
% However, SIS and TARP, which also relies on marginal correlations, still achieve acceptable prediction 
% performance for many active variables in this covariance setting in Figure~\ref{fig:rMSPE_cov_settings}.
% These results also suggest that SPAR behaves more like a dense method. For completeness, Figures~\rom{not anymore, update text} in the appendix 
% show the results for the same covariance settings, but for sparse and dense active variable settings. One interesting result is that 'SPAR 1-se' is the only method with a high F1 score in the dense 'extreme' covariance setting.

Next, we take a closer look at the most general `group' covariance setting with medium active variables and look at the effect of changing $p,n$ or $\rho_\text{snr}$.
Figure~\ref{fig:perf_group_medium} in the Appendix shows that all methods achieve increasingly better performance measures when $p$ is decreasing and that a similar effect can be seen for increasing $n$ and for increasing signal-to-noise ratio $\rho_\text{snr}$, where both versions of SPAR are always among the best methods for prediction. RP\_CW was left out of these Figures because of its worse performance compared to the ensemble version and all other methods.

Finally, Figure~\ref{fig:Ctime} shows the average computing times in the `group' covariance setting. All used methods scale quite well with $p$, where SIS takes the least time to compute, followed by SPAR without cross-validation. One pays the price in terms of computing time when employing the 10-fold cross-validation procedure, where the number of models varies from 10 to 100 and a grid of $20$ $\lambda$ values. However, its slope for increasing $p$ is lower than for most other methods, e.g. PLS. 
It might be surprising that the random projection ensemble takes longer to compute than other methods consisting of more steps, but this is due to the large input dimension of the random projection matrices compared to SPAR and TARP, which use a screening step first to reduce dimension.
% A possibility to speed up the computation to only choose $\lambda$ by CV and fix the number of models $M$. \laura{TODO to what? any recommendation}
% the SPAR method takes similar computing time to well-established methods such as AdLASSO.

\begin{figure}[t!]
  \centering
    \includegraphics[width=0.8\textwidth,trim=0 5 0 10, clip]{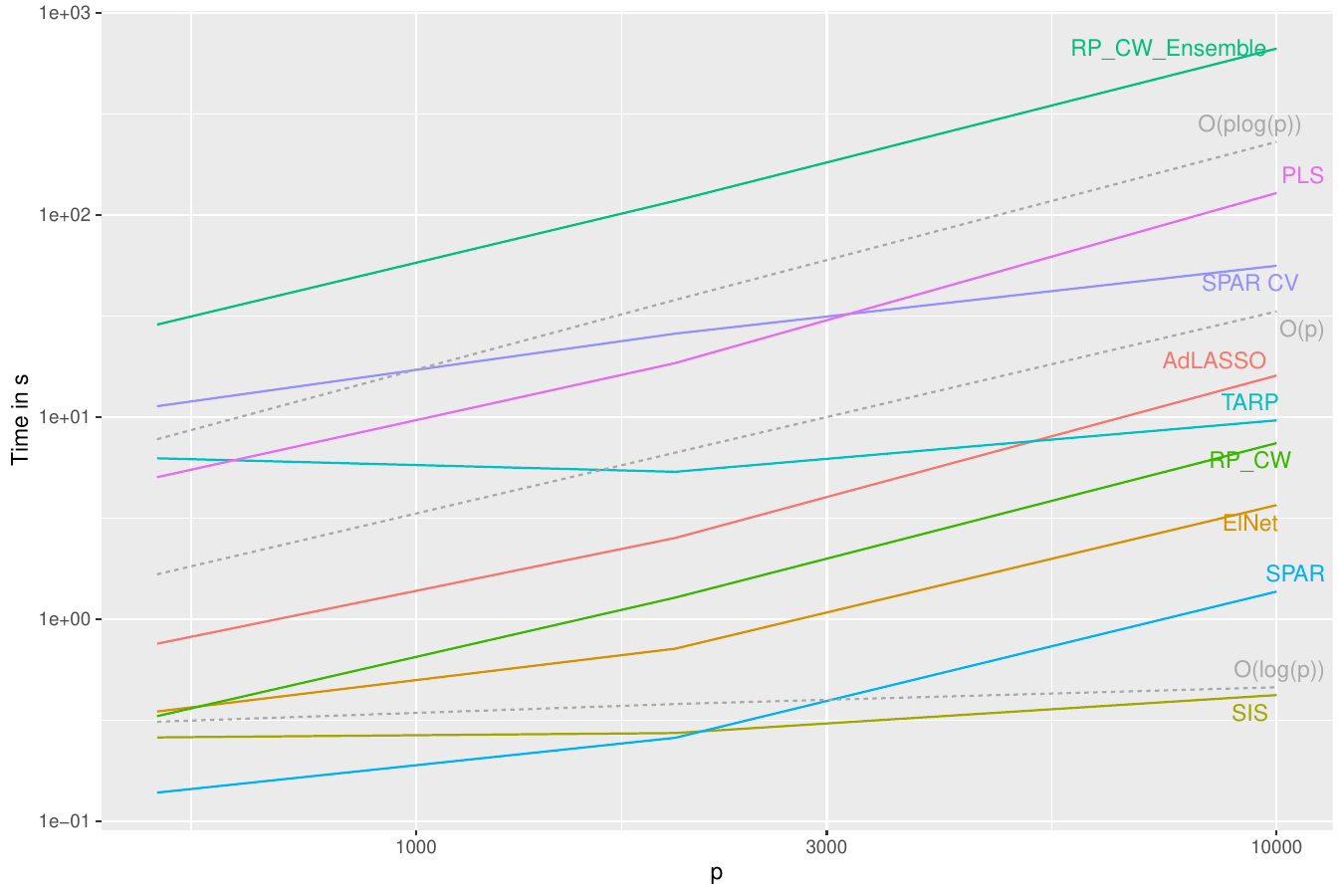}
  \caption{Average computing time in seconds in the `group' covariance setting over $n_\text{rep}=100$ replications of each active variable setting for increasing $p$ and fixed $n=200,\rho_\text{snr}=10$.}
  \label{fig:Ctime}
\end{figure}

\section{Data Applications}\label{sec:data}
In this section, we apply our proposed SPAR method and the same competitors to two real-world high-dimensional regression problems.
% where one could be regarded as a sparse problem and the other one as rather dense.
For both applications, we randomly split the data into training set of size $3n/4$ and test set of size $n/4$ (rounded) for evaluation and 
repeat this process $n_\text{rep}=100$ times.

\subsection{Rateye Gene Expression}\label{sec:data:rat}

In this example, we use the dataset collected by 
%This dataset was obtained for a study by 
\cite{ScheetzEtAl2006RatEye}\footnote{The dataset is publicly available in the Gene Expression Omnibus repository 
\url{www.ncbi.nlm.nih.gov/geo} (GEO assession id: GSE5680)}, where they measured expression levels of $31042$ (non-control) gene probes on collected tissues from eyes of $n=120$ rats. 
% One of these genes, TRIM32, was identified as an additional BBS (Bardet-Biedl syndrome, multisystem human disease) gene \citep{Trim32ScheetzRatEye}.
Similarly to \cite{Huang2006AdLassoSparseRatEye}
we are interested in modeling the relation of all other genes to a specific gene TRIM32, which has been related to Bardet-Biedl syndrome. Since only a few genes are expected to be linked to the given gene, this can be interpreted as a sparse high-dimensional regression problem \citep{Huang2006AdLassoSparseRatEye}. As in \cite{Huang2006AdLassoSparseRatEye}
 and \cite{ScheetzEtAl2006RatEye}, we only use genes that are expressed in the eye and have sufficient variation for our analysis. A gene is expressed, if its maximum observed value is higher than the first quartile of 
all expression values of all genes, and has sufficient variation if it exhibits a coefficient of variation of at least two. 
For us, $p=22905$ filtered genes meet these criteria to be used in the analysis. 
A subset of this dataset with $p=200$ genes is also available in the R-package \texttt{flare} \citep{FlareR}, where all but $3$ genes are also contained in our filtered version. The selection process in \citep{FlareR} is not described in more detail, but all $200$ genes have a higher marginal correlation to TRIM32 than three-quarters of all available genes. 

Figure~\ref{fig:rateye} shows the prediction performance for these two versions of the dataset, where  SPAR CV, SPAR, and TARP perform best on
the bigger dataset. For the small dataset, the TARP and the ensemble of CW random projections achieve the best results. % Interestingly, the sparse methods achieve better performance on the smaller version of the dataset compared to the full filtered dataset, while the others are able to 
% reach better prediction performance from the bigger dataset. 
Both SPAR methods improve their performance when the number of variables increases from 200 to 22905 showing that the method is able to make use of the additional information, while the sparse methods yield worse predictions on the bigger dataset. Note that our used implementation for PLS and PCR failed to work on the bigger dataset, so we show Ridge here instead.

  \begin{figure}[t!]
    \centering
      \includegraphics[width=0.8\textwidth,trim=0 20 0 20, clip]{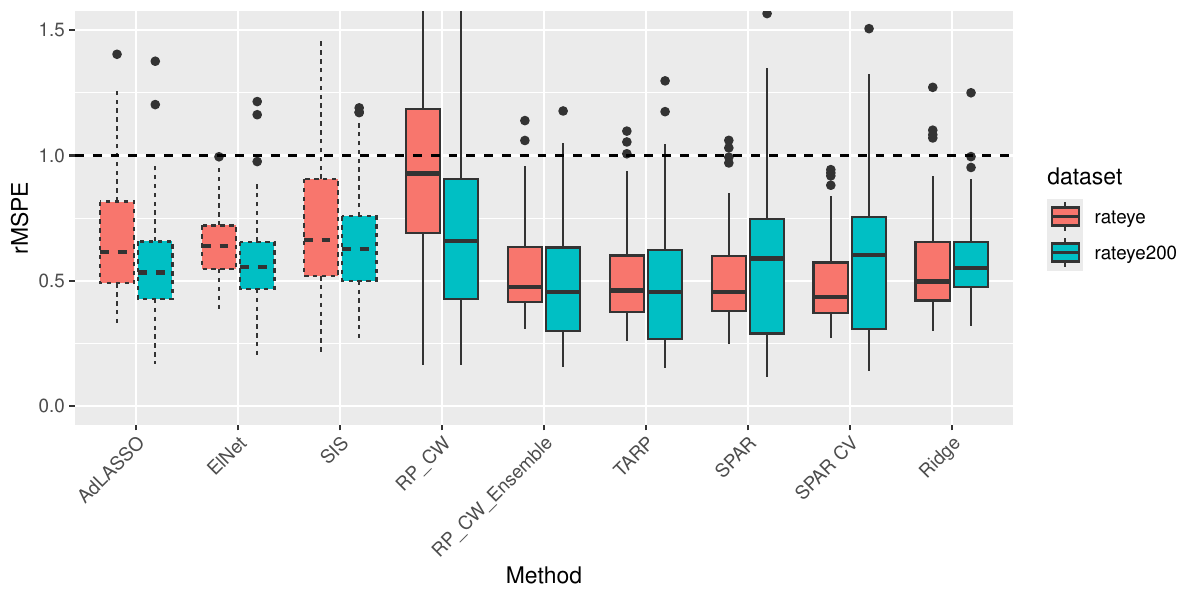}
    \caption{
    Relative MSPEs of the competing methods on the two \textbf{rat eye gene 
    expression datasets} for $n_\text{rep}=100$ random training/test splits.}
    \label{fig:rateye}
  \end{figure}

Table~\ref{tab:medNumAct} shows the median number of active variables of the competing methods on these data applications.
% confirming the simulation results that our SPAR method uses more active variables than sparse methods. 
% In the Appendix Section~\ref{sec:AppC},
% we compare which genes are selected by each of our $ competitors.  
SPAR with $M=20$ and $\lambda=0$ as well as SPAR CV (where $M$ and $\lambda$ are selected by cross-validation)  reduce the predictor space but we can see that the number of used variables is much larger than for the sparse methods. Note that SPAR CV can be less sparse than SPAR, even if it performs additional thresholding, in the cases where the selected number of models used in the ensemble is larger.
% only around $3200$ of the $22905$ genes on average.
% The fact that any sparse method, even just using the one standard error rule instead of the best $\lambda$ in the cross-validation in SPAR for more sparsity,
% achieves worse prediction performance, raises the question whether this problem is actually sparse. 
% Evaluating variable selection in this real-world application, where the truth is unknown, is difficult. 
The fact that the sparse methods do not achieve the best performance on these datasets raises the question 
whether this problem is actually sparse, as in our simulated sparse settings the sparse methods always performed better than the rest. To investigate whether this is due to the observed covariance structure of the predictors rather than the sparsity in $\beta$, we perform a small simulation exercise as in Section~\ref{sec:simulation} where we generate synthetic data with the observed predictors. Results are presented in Appendix Section~\ref{sec:AppC}. We show that, for such a covariance structure, the SPAR method performs well even in the sparse setting, leaving the question of the true sparsity in this problem open. 
 
\subsection{Face Images}\label{sec:data:face}

The dataset originates from \cite{Tenenbaum2000ISOMAPFaceData}\footnote{\textit{Isomap face data} can be found online on \url{https://web.archive.org/web/20160913051505/http://isomap.stanford.edu/datasets.html}}
and was also studied, among others, in \cite{Dunson2016ComprGP}. It consists of $n=698$ black and white face images 
of size $p=64\times 64=4096$ and the faces' horizontal looking direction angle as the response. The bottom left plot in Figure~\ref{fig:faces_vis} illustrates
one such instance with the corresponding angle. For each training/test split, we exclude pixels close to the edges and corners, which are constant on the training set.  We expect that in the linear model, many pixels together carry relevant information, making this a rather dense regression problem.
 
Figure~\ref{fig:rMSPE_faces} shows the prediction performance results for this dataset. Here,  SPAR and SPAR CV yield 
the lowest prediction error, followed by PLS. AdLASSO and SIS perform substantially worse than the other methods, as their number of estimated active predictors
seems to be too low, see Table~\ref{tab:medNumAct}.
It should be noted that this example was previously used for illustrating non-linear methods in (low-dimensional) manifold regression in \cite{Dunson2016ComprGP}. When replicating the preprocessing 
% with additive input noise $\tau=0.03$ 
in their paper, SPAR CV achieved an average MSPE of $0.0142$ with average bootstrap standard error of $0.0043$, while the best non-linear method mentioned in \cite{Dunson2016ComprGP} achieved $0.06$ with standard error $0.009$, showing that the proposed method with the linear model assumption is a feasible option for modeling this data.
  \begin{figure}[t!]
    \centering
    \includegraphics[width=0.8\columnwidth,trim=0 20 0 10,clip]{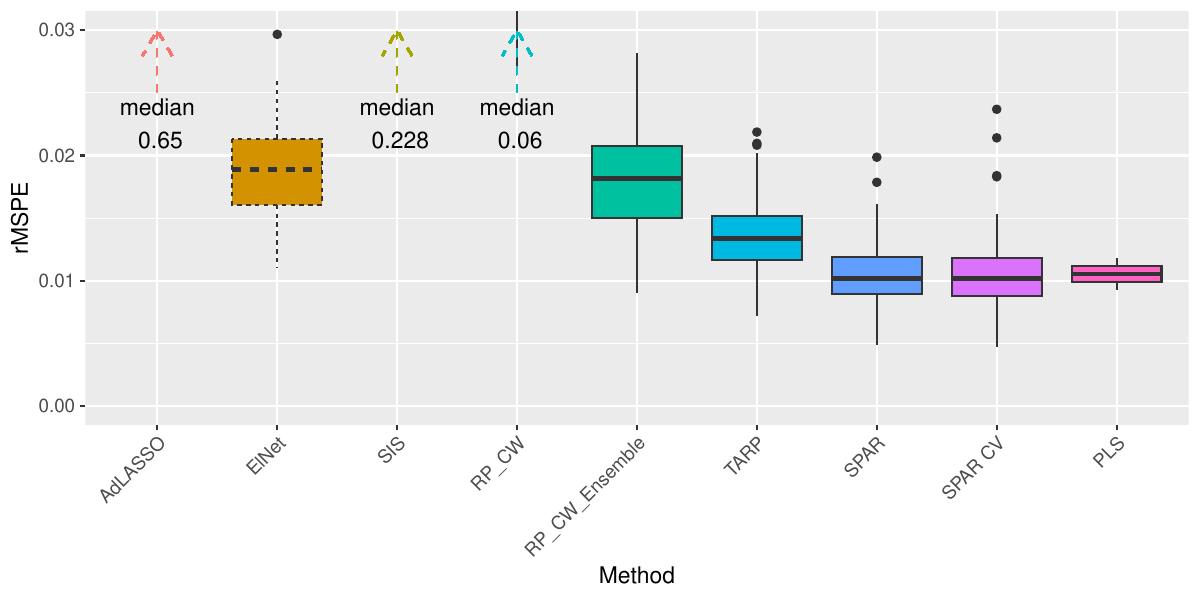}
    \caption{Relative MSPE of the competing methods, where the sparse methods are marked by dotted lines, on the \textbf{face angle dataset} for $n_\text{rep}=100$ random training/test splits}
    \label{fig:rMSPE_faces}
  \end{figure}

For this dataset, we also illustrate the estimated regression coefficients and their contribution to a new prediction for SPAR CV.
We apply our method once on the full dataset except for two test images, thus $n=696$.
The top of Figure~\ref{fig:faces_vis} shows the positive (left) and negative (right) estimated regression coefficients of the pixels. 
It yields almost symmetrical images, which is sensible, and highlights the contours of the nose and forehead. For the prediction of a new face image, 
we can define the contribution of each pixel as the pixel's regression coefficient multiplied by the corresponding pixel grey-scale value of the new instance.
In the bottom right we visualize these contributions for the test instance on the bottom left. 
The sum of all these contributions (plus a `hidden' intercept) yields the prediction of $\hat y = 34.8$ for the true angle $y=35.2$.

  \begin{figure}[t!]
    \centering
    \includegraphics[width=0.8\textwidth,trim=0 20 0 20, clip]{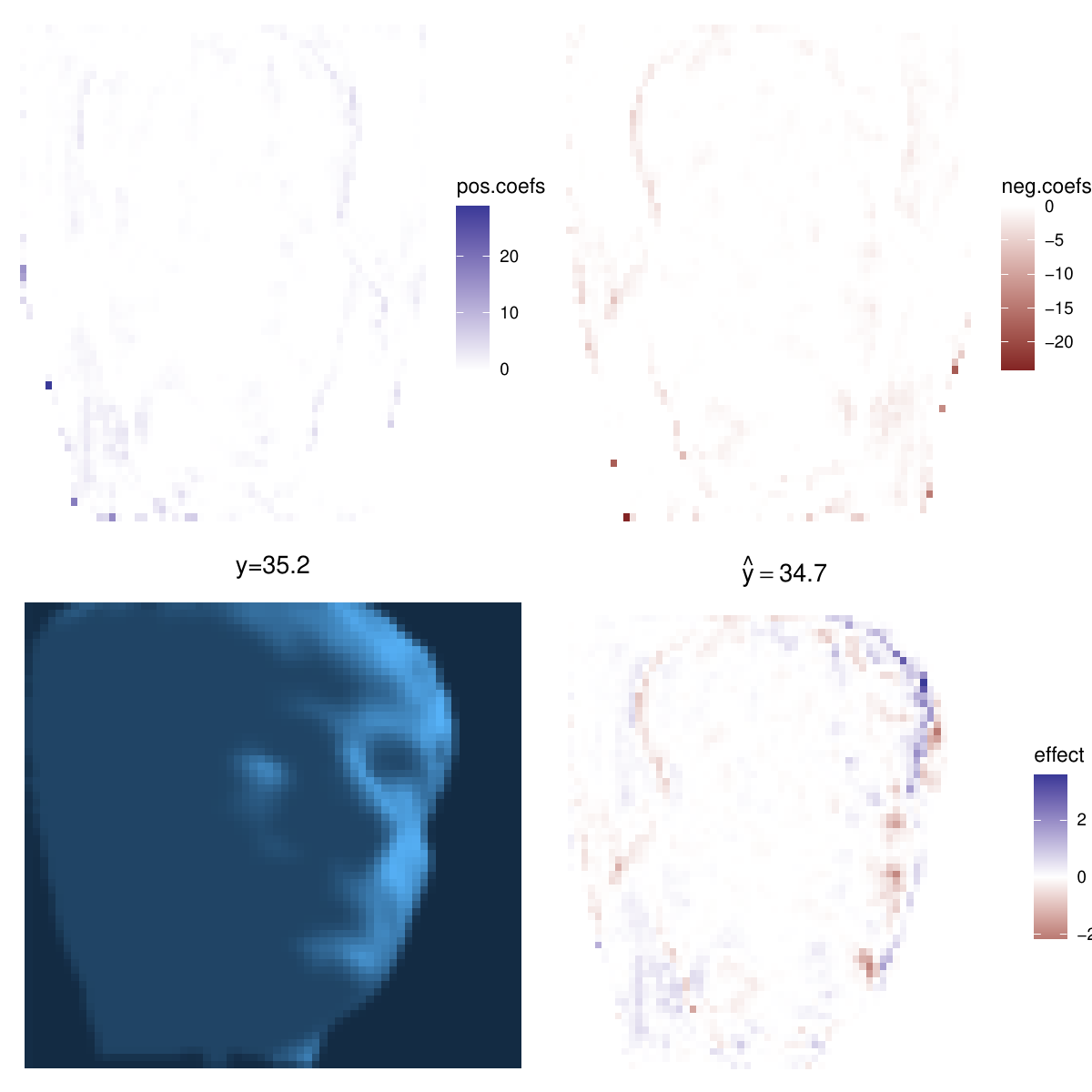}
    \caption{Top: positive (left) and negative (right) estimated coefficients of SPAR CV.
    Bottom: One new instance (left) and the contributions of each pixel to its prediction (right)
    }
    \label{fig:faces_vis}
  \end{figure}

 % \rom{Figure removed} Figure illustrates the threshold selection of our SPAR method for this application, 
 % where we see the mean squared error estimated by cross-validation and the corresponding error band on the top, and the implied number of active variables on the bottom, over 
 % a grid of $\lambda$-values (displayed as rounded to three digits). With almost the same estimated prediction error, the `SPAR 1-se' method uses over $1000$ pixels less than `SPAR best'. 
 % For the previous data application, the difference is even more severe, as shown in Table~\ref{tab:medNumAct}.

\begin{table}[t!]
    \caption{Median number of active predictors for all methods on data applications 
    across $n_\text{rep}=100$ random training/test splits}
    \label{tab:medNumAct}
    \centering
\begin{tabular}{llll}
\toprule
Method & rateye & rateye200 & face\\
\midrule
$p$ & 22905.0 & 200.0 & 3890.0\\
AdLASSO & 44.0 & 10.0 & 19.5\\
ElNet & 24.5 & 18.5 & 113.5\\
SIS & 5.0 & 4.0 & 5.5\\
\addlinespace
TARP & 16490.0 & 200.0 & 3626.5\\
SPAR & 3203.5 & 199.0 & 3487.0\\
SPAR CV & 6797.0 & 194.0 & 3790.0\\
\bottomrule
\end{tabular}
\end{table}

\section{Summary and Conclusions}\label{sec:conclusion}

In this paper, we introduced a new data-informed random projection aimed at dimension reduction for linear regression, which 
uses the HOLP estimator \citep{Wang2015HOLP} from variable screening literature, together with a theoretical result showing how much 
better we can expect the prediction error to be compared to a conventional random projection.

Around this new random projection, we built the SPAR ensemble method with a data-driven threshold selection using cross-validation.
% We propose two different choices for this threshold. Firstly, the value providing the smallest cross-validated MSE, and secondly, the
% value leading to the sparsest coefficient while still achieving a similar cross-validated MSE. The first one should be chosen when purely predictions are of interest. In case we want to interpret the model and identify important variables, the second version should be prefered.
In an extensive simulation study we compare SPAR to different methods employing random projections as well as with other sparse and dense methods we show that the proposed method achieves 
the best all-around prediction as well as variable ranking performance across the scenarios, proving that it is a viable option, especially in cases where it is not clear how sparse the problem is in practice.  
We also notice that SPAR's prediction performance is comparable to the cross-validated SPAR CV, however, SPAR CV shows better performance in terms of variable ranking. In applications where computationally fast prediction is the main task, SPAR without cross-validation can be employed without notable loss in prediction ability.

% While the performance of the different methods highly depends on the investigated scenario, with sparse methods performing best in most investigated sparse settings, we show that SPAR is most competitive across many different settings.
% However, as shown in our simulations, in very sparse settings, the data and MSE-driven threshold selection leads to too many active variables, and sparse methods 
% end up performing better both for prediction and variable selection. 

% Another open problem is the high variety that comes with the random SPAR method. 
% Sound theory to quantify the uncertainty of random projections and, especially, ensembles of them, is still lacking.

This methodology can be extended to non-linear (or robust) regression by employing non-linear (or robust) methods, 
such as generalized linear models or Gaussian processes, in the marginal models instead of OLS.
Possible future work also includes extending the work for classification tasks.

Finally, we saw that variable selection remains a difficult task in high-dimensional regression for all considered methods. Our proposed method SPAR exhibits a low precision and high recall in sparse settings, showing that it tends to use too many variables in the procedure (see Figures~\ref{fig:Prec_cov_settings} and ~\ref{fig:Rec_cov_settings} in Appendix). How to modify the method to detect the right degree of sparsity remains an open problem, if such a degree can even be determined in real-world problems.

\paragraph*{Computational Details}
All code to reproduce the simulation study and all other experiments, Figures and Tables can be found in the GitHub repository \url{https://github.com/RomanParzer/SPAR_Paper_Figures_Code}.

\paragraph*{Declarations}

Roman Parzer and Laura Vana-Gür acknowledge funding from the Austrian Science Fund (FWF) for
the project ``High-dimensional statistical learning: New methods to advance
economic and sustainability policies'' (ZK 35), jointly carried out by WU
Vienna University of Economics and Business, Paris Lodron University
Salzburg, TU Wien, and the Austrian Institute of Economic Research (WIFO).

% \bibliographystyle{unsrtnat}
% \bibliography{../lit.bib}

\clearpage

\bigskip
\begin{center}
{\large\bf SUPPLEMENTARY MATERIAL}
\end{center}

\begin{appendices}

\section{Lemmas and Proof of Theorem~\ref{theorem:1}}\label{sec:AppA}

This section states and proves Lemmas~\ref{lemma:altRidge},~\ref{lemma:HOLPminNorm}~and~\ref{lemma:ProjPhi} mentioned in 
Section~\ref{sec:method}, and gives a detailed proof of Theorem~\ref{theorem:1} and Lemma~\ref{lemma:ExpBinomRatio} needed in the proof.

\begin{lemma}\label{lemma:altRidge}
    Let $X\in\mathbb{R}^{n\times p}$ be a fixed matrix and $y\in\mathbb{R}^n$ a vector. Then, the Ridge estimator for $\lambda>0$ has the following alternative form suitable for the $p\gg n$ case.
    \begin{align}
      \hat\beta_\lambda &:= (X'X + \lambda I_p)^{-1} X'y = X'(\lambda I_n + XX')^{-1} y
    \end{align}
  \end{lemma}
  \begin{proof}
    Using the Woodbury matrix inversion formula 
  \begin{align*}
    (A+UCV)^{-1} = A^{-1} - A^{-1} U(C^{-1} + VA^{-1} U)^{-1} VA^{-1}, 
  \end{align*}
  where $A$, $U$, $C$ and $V$ are conformable matrices,
  we have for any penalty $\lambda>0$
  \begin{align*}
    \hat\beta_\lambda &:= (X'X + \lambda I_p)^{-1} X'y = \\
    &= \frac{1}{\lambda}\Big(I_p - 1/\lambda \cdot X'(I_n + 1/\lambda \cdot XX')^{-1} X\Big) X'y = \\
    &= \frac{1}{\lambda}X'y - \frac{1}{\lambda} X'(\lambda I_n + XX')^{-1} XX'y \pm \frac{1}{\lambda} X'(\lambda I_n + XX')^{-1} \lambda y = \\
    &= \frac{1}{\lambda}X'y - \frac{1}{\lambda} X'\underbrace{(\lambda I_n + XX')^{-1} (XX' +\lambda I_n)}_{=I_n} y + \frac{1}{\lambda} X'(\lambda I_n + XX')^{-1} \lambda y = \\
    &= X'(\lambda I_n + XX')^{-1} y.
  \end{align*}
  \end{proof}
  
  \begin{lemma}\label{lemma:HOLPminNorm}
    Let $X\in\mathbb{R}^{n\times p}$ be a fixed matrix with $\text{rank}(XX')=n$ (implying $p>n$) and $y\in\mathbb{R}^n$ a vector. Then, 
    the minimum norm least-squares solution $\text{argmin}_{\beta \in\mathbb{R}^p,s.t. X\beta=y}\|\beta\|$ is uniquely given by
    $\hat \beta=X'(XX')^{-1} y$.
  \end{lemma}
  \begin{proof}
    Obviously, $\hat \beta=X'(XX')^{-1} y$ satisfies $X\hat\beta=y$. For any $\tilde \beta \in\mathbb{R}^p$ with $X\tilde\beta=y$  we have
    \begin{align*}
      \|\tilde \beta \|^2 &= \|\hat \beta + \tilde \beta - \hat \beta \|^2 = \|\hat \beta \|^2 + \|\tilde \beta - \hat \beta \|^2 + 2\cdot \hat \beta'(\tilde \beta - \hat \beta)  = \\
      &= \|\hat \beta \|^2 + \underbrace{\|\tilde \beta - \hat \beta \|^2}_{\geq0} + 2\cdot y'(XX')^{-1} \underbrace{X(\tilde \beta - \hat \beta)}_{=0} \geq \|\hat \beta \|^2,
    \end{align*}
    with equality if and only if $\tilde \beta = \hat \beta$.
  \end{proof}
  
  \begin{lemma}\label{lemma:ProjPhi}
    Let $\Phi\in\mathbb{R}^{m\times p}$ be a CW random projection from Definition~\ref{def:RP_CW} with general diagonal elements $d_j\in\mathbb{R}$ and $\beta\in\mathbb{R}^p$.
    Then, the projected vector $\tilde \beta = P_\Phi \beta$ for the orthogonal projection $P_\Phi=\Phi'(\Phi\Phi')^{-1}\Phi$ onto the row-span of $\Phi$ is given by
    \begin{align}
      \tilde \beta_j = d_j \cdot \frac{\sum_{k:h_k=h_j}d_k\beta_k}{\sum_{k:h_k=h_j}d_k^2}.
    \end{align}
  \end{lemma}
  \begin{proof}
    We can split the projection in
    \begin{align*}
      P_\Phi\beta =\Phi'(\Phi\Phi')^{-1}\Phi \beta = D(B'(\Phi\Phi')^{-1} B) (D\beta).
    \end{align*}
    The matrix $\Phi\Phi' = BD^2B' \in\mathbb{R}^{m\times m}$ is diagonal with entries $\{\sum_{l:h_l=i}d_l^2:i\in[m]\}$, because each variable is only mapped to one goal dimension.
    Then, for $j,k\in[p]$ we have
    \begin{align*}
      (B'(\Phi\Phi')^{-1} B)_{jk}= \begin{cases}
        0 & h_j\neq h_k \\
        1/(\sum_{l:h_l=h_j}d_l^2) & h_j = h_k
      \end{cases}.
    \end{align*}
    Putting it together, we get
    \begin{align*}
      \tilde \beta_j = d_j \cdot \sum_{k=1}^p I\{h_k=h_j\} \cdot \frac{d_k\beta_k}{\sum_{l:h_l=h_j}d_l^2}= d_j \cdot \frac{\sum_{k:h_k=h_j}d_k\beta_k}{\sum_{k:h_k=h_j}d_k^2}.
    \end{align*}  
  \end{proof}

  \begin{lemma}\label{lemma:ExpBinomRatio}
    Let $h:[p]\to[m]$ be a random map such that for each $j\in[p]: h(j)=h_j \overset{iid}{\sim} \text{Unif}([m])$, and let 
    $\mathcal{A}\subset [p]$ be a subset of indizes with $a=|\mathcal{A}|>1$.
    Then,
    \begin{align}
      \mathbb{E}[\frac{| \mathcal{A} \cap h^{-1}(h_j)\setminus \{j\}|}{|h^{-1}(h_j)|}] &= \frac{a-\mathbf{I}\{j\in\mathcal{A}\}}{p-1}\cdot\Bigl(1-\frac{m}{p}(1-(\frac{m-1}{m})^p)\Bigr), \\
      \mathbb{E}[\frac{| \mathcal{A} \cap h^{-1}(h_j)\setminus \{j\}|}{|h^{-1}(h_j)|^2}] &= m\frac{a-\mathbf{I}\{j\in\mathcal{A}\}}{p-1}\cdot\Bigl(\frac{1}{p-1} - \frac{m+1}{(p-1)^2} + \mathcal{O}(p^{-3})\Bigr),
    \end{align}

    where $h^{-1}(k)=\{j\in[p]:h(j)=k\}$ is the (random) preimage set for $k\in[m]$.
    \end{lemma}
    \begin{proof}
      The first random variable $| \mathcal{A} \cap h^{-1}(h_j)\setminus \{j\}|/|h^{-1}(h_j)|$ (random in $h$) has the distribution of $X_1/(1+X_1+X_2)$, where 
      $X_1\sim \text{Binom}(a_j,1/m), a_j = a - \mathbf{I}\{j\in\mathcal{A}\}$ corresponding to the active variables 
      (except $j$) and $X_2\sim\text{Binom}(p-1-a_j,1/m)$ independent of $X_1$ corresponding to the inactive variables.
      
      Note that for any $x_1,x_2\in\mathbb{N}$ $x_1/(1+x_1+x_2) = \int_0^1x_1s^{x_1+x_2} ds$
      and, by Fubini's theorem, we can interchange the integral and expectation to obtain
      \begin{align*}
        \mathbb{E}[\frac{X_1}{1+X_1+X_2}] = \int_0^1\mathbb{E}[X_1s^{X_1}]\mathbb{E}[s^{X_2}] ds.
      \end{align*}
      By using the moment-generating-function of a binomial variable and the dominated convergence theorem to interchange the derivative and the expectation, we get 
      \begin{align*}
        \mathbb{E}[s^{X_2}] & = \Bigl(\frac{m-1}{m} + \frac{1}{m}s\Bigr)^{p-1-a_j},\\
        \mathbb{E}[(X_1+1)s^{X_1}] &= \frac{\partial}{\partial s}\mathbb{E}[s^{X_1+1}] =\frac{\partial}{\partial s} s \Bigl(\frac{m-1}{m} + \frac{1}{m}s\Bigr)^{a_j} = \\
        &= \Bigl(\frac{m-1}{m} + \frac{1}{m}s\Bigr)^{a_j} + s \frac{a_j}{m}\Bigl(\frac{m-1}{m} + \frac{1}{m}s\Bigr)^{a_j-1}, \\
        \implies \mathbb{E}[X_1s^{X_1}] &= \mathbb{E}[(X_1+1)s^{X_1}] - \mathbb{E}[s^{X_1}] = s \frac{a_j}{m}\Bigl(\frac{m-1}{m} + \frac{1}{m}s\Bigr)^{a_j-1}.
      \end{align*}
    Putting the results together and using partial integration, we obtain
      \begin{align*}
        \mathbb{E}[\frac{X_1}{1+X_1+X_2}] &= \int_0^1 s \frac{a_j}{m}\Bigl(\frac{m-1}{m} + \frac{1}{m}s\Bigr)^{a_j-1} \Bigl(\frac{m-1}{m} + \frac{1}{m}s\Bigr)^{p-1-a_j} ds = \\
        &= \frac{a_j}{p-1}\cdot\Bigl(1-\frac{m}{p}(1-(\frac{m-1}{m})^p)\Bigr).
      \end{align*}
    
    Similarly, the second random variable $| \mathcal{A} \cap h^{-1}(h_j)\setminus \{j\}|/|h^{-1}(h_j)|^2$ has the distribution of $X_1/(1+X_1+X_2)^2$.
    We will use a similar approach to \citet{CribariNeto2000InvMomentsBinom} to obtain a fourth-order approximation.
    
    By use of the Gamma-function and similar arguments to the first case, we can write
    \begin{align*}
      \frac{x_1}{(1+x_1+x_2)^2} &= \int_0^\infty x_1t e^{-(1+x_1+x_2)t}dt
    \end{align*}
    for any $x_1,x_2\in\mathbb{N}$, and
    \begin{align}\label{eqn:ExpBinRat2}
      \mathbb{E}[\frac{X_1}{(1+X_1+X_2)^2}] &= \int_0^\infty te^{-t} \mathbb{E}[X_1e^{-X_1t}] \mathbb{E}[e^{-X_2t}] dt.
    \end{align}
    By use of the moment-generating-functions we get 
    \begin{align*}
      \mathbb{E}[e^{-X_2t}] &= \Bigl(\frac{m-1}{m} + \frac{1}{m} e^{-t}\Bigr)^{p-1-a_j},\\
      \mathbb{E}[X_1e^{-X_1t}] &= \mathbb{E}[\frac{\partial}{\partial t}\Bigl( - e^{-X_1t}\Bigr)] = -\frac{\partial}{\partial t} \mathbb{E}[e^{-X_1t}] = -\frac{\partial}{\partial t} \Bigl(\frac{m-1}{m} + \frac{1}{m} e^{-t}\Bigr)^{a_j} = \\
      &= a_j\Bigl(\frac{m-1}{m} + \frac{1}{m}e^{-t}\Bigr)^{a_j-1} \frac{1}{m} e^{-t}.
    \end{align*}
    
    Plugging this into \eqref{eqn:ExpBinRat2} and using the variable substitution $e^{-r} = \frac{m-1}{m} + \frac{1}{m}e^{-t}$ and the definition $g(r) = -\log(m(e^{-r}-\frac{m-1}{m})) me^{-r}$ yields
    \begin{align}
      \mathbb{E}[\frac{X_1}{(1+X_1+X_2)^2}] &= a_j \int_0^\infty \frac{1}{m} te^{-2t} \Bigl(\frac{m-1}{m} + \frac{1}{m}e^{-t}\Bigr)^{p-2} dt = \nonumber \\
      &= a_j \int_0^{-\log(\frac{m-1}{m})} -\log(m(e^{-r}-\frac{m-1}{m})) m(e^{-r} - \frac{m-1}{m}) e^{-(p-1)r} dr =  \nonumber \\
      &= a_j\int_0^{-\log(\frac{m-1}{m})} \Bigl(1- \frac{m-1}{m}e^r\Bigr)g(r)  e^{-(p-1)r} dr. \label{eqn:a16}
    \end{align}
    
    From \citet{CribariNeto2000InvMomentsBinom} we use the facts that for $\delta<\min(1,-\log(\frac{m-1}{m}))$
    \begin{align}
      g(r) &= m^2 r \Bigl[1 + \frac{m-3}{2} r + \mathcal{O}(r^2)\Bigr], \\
       \int_0^\delta r^k e^{-(p-1)r}dr &= \frac{\Gamma(k+1)}{(p-1)^{k+1}} +\mathcal{O}(e^{-(p-1)\delta}),\\
    \int_\delta^{-\log(\frac{m-1}{m})}g(r) e^{-(p-1)r}dr &= \mathcal{O}(e^{-(p-1)\delta}).
    \end{align}
    On $r<\delta$ we use the Taylor expansion $e^r = 1 + r + \mathcal{O}(r^2)$ to obtain from \eqref{eqn:a16}
      \begin{align*}
        \mathbb{E}[\frac{X_1}{(1+X_1+X_2)^2}] &= a_j\Bigl[m^2 \int_0^{\delta} \Bigl(r\frac{1}{m} + r^2(-\frac{m-1}{m} + \frac{m-3}{2m}) + \mathcal{O}(r^3) \Bigr) e^{-(p-1)r} dr + \\
        & \quad\quad \mathcal{O}(e^{-(p-2)\delta})\Bigr] =\\
        &= a_j m \Bigl[\frac{1}{(p-1)^2} + \frac{2(-(m-1)+(m-3)/2)}{(p-1)^3} + \mathcal{O}(p^{-4})\Bigr]  = \\
        &= m\frac{a_j}{p-1}\cdot\Bigl(\frac{1}{p-1} - \frac{m+1}{(p-1)^2}+ \mathcal{O}(p^{-3}) \Bigr) .
      \end{align*}
      
    \end{proof}

  \begin{proof}[Proof of Theorem~\ref{theorem:1}]\label{proof:Th1}
    For a general CW projection $\Phi=BD$, reduced predictors $Z=X\Phi'$, and a prediction $\hat y = (\Phi\tilde x)'(Z'Z)^{-1} Z'y = (\Phi\tilde x)'(Z'Z)^{-1} Z'X\beta + (\Phi\tilde x)'(Z'Z)^{-1} Z'\varepsilon$
    we get the expected squared error (w.r.t $\tilde x, \tilde \varepsilon$, and $\varepsilon$ given $X$ and $\Phi$)
    \begin{align}
      \mathbb{E}[(\tilde y - \hat y)^2|X,\Phi] &= \mathbb{E}[\Big(\tilde x'(I_p -  \Phi'(Z'Z)^{-1} Z'X)\beta + \tilde\varepsilon - \tilde x' \Phi'(Z'Z)^{-1} Z'\varepsilon \Big)^2|X,\varepsilon,\Phi] = \\
      &=  \mathbb{E}[ \beta' (I_p -  X'X\Phi'(\Phi X'X\Phi')^{-1}\Phi) \tilde x \tilde x'(I_p -  \underbrace{\Phi'(\Phi X'X\Phi')^{-1} \Phi X'X}_{:=P})\beta  \\
      &\quad + \tilde\varepsilon^2 + \varepsilon' X\Phi'(\Phi X'X\Phi')^{-1}\Phi\tilde x \tilde x'\Phi'(\Phi X'X\Phi')^{-1} \Phi X'\varepsilon|X,\Phi] = \\
      &= \beta'(I_p-P)'\Sigma(I_p-P)\beta+ \sigma^2 \label{eqn:condMSE_prTh1} \\
      &\quad + \mathbb{E}[\varepsilon' X\Phi'(\Phi X'X\Phi')^{-1}\Phi\Sigma\Phi'(\Phi X'X\Phi')^{-1} \Phi X'\varepsilon|X,\Phi],\label{eqn:condMSE_prTh2}
    \end{align}
    where we used that the mixed terms have expectation $0$. The third term has conditional expectation given $\Phi$
    \begin{align*}
      \mathbb{E}[\varepsilon' &X\Phi'(\Phi X'X\Phi')^{-1}\Phi\Sigma\Phi'(\Phi X'X\Phi')^{-1} \Phi X'\varepsilon|\Phi] = \\
      &= \mathbb{E}[\text{tr}\Big( (\Phi X'X\Phi')^{-1}\Phi\Sigma\Phi'(\Phi X'X\Phi')^{-1} \Phi X'\varepsilon\varepsilon' X\Phi' \Big)|\Phi] = \\
      &= \sigma^2 \cdot \text{tr}\Big( \mathbb{E}[(\Phi X'X\Phi')^{-1}|\Phi]\Phi\Sigma\Phi'\Big),
    \end{align*}
    where we used the facts that $\text{tr}(AB)=\text{tr}(BA)$ for matrices $A,B$ of suitable dimensions, 
    $\mathbb{E}[\varepsilon\varepsilon'] = \sigma^2 \cdot I_n$ and $\varepsilon$ is independent of $X$ and $\Phi$. 
    For fixed $\Phi$, the matrix $X\Phi'$ has a centered matrix normal distribution with among-row covariance $I_n$ and among-column covariance $\Phi\Sigma\Phi'\in\mathbb{R}^{m\times m}$.
    Therefore, $\Phi X'X\Phi'$ has a Wishart distribution with scale matrix $\Phi\Sigma\Phi'\in\mathbb{R}^{m\times m}$ and $n$ degrees of freedom, and 
    $(\Phi X'X\Phi')^{-1}$ has an Inverse-Wishart distribution resulting in the expectation $\mathbb{E}[(\Phi X'X\Phi')^{-1}|\Phi] = (\Phi\Sigma\Phi')^{-1} / (n-m-1)$ and, continuing above
  calculations, we obtain
  \begin{align*}
    \mathbb{E}[\varepsilon' &X\Phi'(\Phi X'X\Phi')^{-1}\Phi\Sigma\Phi'(\Phi X'X\Phi')^{-1} \Phi X'\varepsilon] = \sigma^2 \cdot \frac{m}{n-m-1}.
  \end{align*}
  Since the expectations of the second and third term in \eqref{eqn:condMSE_prTh1} and \eqref{eqn:condMSE_prTh2} do not depend on $\Phi$ or the respective diagonal elements, they
  will cancel when computing the difference in \eqref{eqn:TH1} and we only need to consider the first term 
  $\beta'(I_p-P)'\Sigma(I_p-P)\beta=(\beta-P\beta)'\Sigma(\beta-P\beta)$. The plan is to find an upper bound on its expectation when
  using diagonal elements proportional to the true coefficient and a lower bound when using random signs as the diagonal elements.
  
  \textit{Lower bound for random signs:}
  Let $\lambda_1\geq\dots\geq\lambda_p>0$ be the ordered eigenvalues of $\Sigma$ and $P_X^{\text{rs}} = \Phi_{\text{rs}}'(\Phi_{\text{rs}} X'X\Phi_{\text{rs}}')^{-1} \Phi_{\text{rs}} X'X$.
  Then,
  \begin{align}\label{eqn:LB_PrTh1}
    \mathbb{E}[(\beta-P_X^{\text{rs}}\beta)'\Sigma(\beta-P_X^{\text{rs}}\beta)] \geq \lambda_p\cdot \mathbb{E}[ \| \beta- P_X^{\text{rs}}\beta \|^2].
  \end{align}
  Let $P_\Phi^{\text{rs}} = \Phi_{\text{rs}}'(\Phi_{\text{rs}}\Phi_{\text{rs}}')^{-1}\Phi_{\text{rs}}$ and $\tilde \beta^{\text{rs}} = P_\Phi^{\text{rs}}\beta$ be the orthogonal projection. Then, we have
  \begin{align*}
    \|\beta- P^{\text{rs}}_X\beta\|^2 &= \|\beta-\tilde \beta^{\text{rs}}\|^2 + \underbrace{\|\tilde\beta^{\text{rs}}- P^{\text{rs}}_X\beta\|^2}_{\geq 0} \geq  \|\beta-\tilde \beta^{\text{rs}}\|^2,
  \end{align*}
  because $\tilde\beta^{\text{rs}}- P^{\text{rs}}_X\beta \in \text{span}(\Phi_{\text{rs}}')$ and $\beta-\tilde \beta^{\text{rs}} \perp \text{span}(\Phi_{\text{rs}}')$.
  
  Using the explicit form of $\tilde \beta^{\text{rs}}$ from Lemma~\ref{lemma:ProjPhi} and independence of the map $h$ and diagonal elements $d_j\overset{iid}{\sim} \text{Unif}(\{-1,1\})$, we get
  \begin{align}\label{eqn:E_tildebeta_rs}
    \mathbb{E}[\tilde \beta_j^{\text{rs}}] &= \mathbb{E}[d_j\cdot \frac{\sum_{k:h_k=h_j}d_k\beta_k}{|h^{-1}(h_j)|}] = \beta_j \cdot \mathbb{E}[\frac{1}{|h^{-1}(h_j)|}].
  \end{align}
  Since we always have $j\in h^{-1}(h_j)$ and the other goal dimensions are independently drawn uniformly at random, the cardinality of 
  this set has distribution $|h^{-1}(h_j)| \sim 1 + \text{Binom}(p-1,1/m)$. \citet{CribariNeto2000InvMomentsBinom} showed that the inverse moments are then given by
  \begin{align*}
    \mathbb{E}[\frac{1}{|h^{-1}(h_j)|}] &= \frac{m}{p}(1-(\frac{m-1}{m})^p), \\
    \mathbb{E}[\frac{1}{|h^{-1}(h_j)|^2}] &= \frac{m^2}{(p-1)^2} + \frac{(m-3)m^2}{(p-1)^3} + \mathcal{O}(p^{-4}).
  \end{align*}
  Plugging this into \eqref{eqn:E_tildebeta_rs} yields
  \begin{align*}
    \beta_j\mathbb{E}[\tilde \beta_j^{\text{rs}}] &= \beta_j^2 \cdot \frac{m}{p}(1-(\frac{m-1}{m})^p)\leq \beta_j^2 \cdot \frac{m}{p} , \\
    \mathbb{E}[(\tilde \beta_j^{\text{rs}})^2|h] &= \mathbb{E}[\frac{\sum_{k:h_k=h_j}\sum_{l:h_l=h_j}d_k d_ld_j^2\beta_k\beta_l }{|h^{-1}(h_j)|^2}|h] = \\
    &= \frac{\sum_{k:h_k=h_j}\beta_k^2 }{|h^{-1}(h_j)|^2} \geq \tau^2 \frac{| \mathcal{A} \cap h^{-1}(h_j)|}{|h^{-1}(h_j)|^2},
  \end{align*}
  where $\tau=\min_{j:\beta_j\neq0}|\beta_j|$. Using Lemma~\ref{lemma:ExpBinomRatio} we get for $\beta_j\neq 0$ (or $j\in\mathcal{A}$)
  \begin{align*}
    \mathbb{E}[(\tilde \beta_j^{\text{rs}})^2] &\geq \tau^2 \mathbb{E}[\frac{| \mathcal{A} \cap h^{-1}(h_j)|}{|h^{-1}(h_j)|^2}] = \tau^2 \mathbb{E}[\frac{1+| \mathcal{A} \cap h^{-1}(h_j)\setminus \{j\}|}{|h^{-1}(h_j)|^2}] = \\
    &= \tau^2\Bigl[\frac{m^2}{(p-1)^2} + \frac{(m-3)m^2}{(p-1)^3} + \mathcal{O}(p^{-4}) + \\
    & \qquad \quad  m\frac{a-1}{p-1}\cdot\Bigl(\frac{1}{p-1} - \frac{m+1}{(p-1)^2}  + \mathcal{O}(p^{-3})\Bigr)\Bigr] \geq \\
    &\geq \tau^2\Bigl[ m\frac{a}{p-1}\cdot\Bigl(\frac{1}{p-1} - \frac{m+1}{(p-1)^2} + \mathcal{O}(p^{-3}) \Bigr)\Bigr]
  \end{align*}
  and, for $\beta_j= 0$ (or $j\notin\mathcal{A}$)
  \begin{align*}
    \mathbb{E}[(\tilde \beta_j^{\text{rs}})^2] &\geq \tau^2 \mathbb{E}[\frac{| \mathcal{A} \cap h^{-1}(h_j)|}{|h^{-1}(h_j)|^2}] = \tau^2 \mathbb{E}[\frac{| \mathcal{A} \cap h^{-1}(h_j)\setminus \{j\}|}{|h^{-1}(h_j)|^2}] = \\
    &= \tau^2\Bigl[ m\frac{a}{p-1}\cdot\Bigl(\frac{1}{p-1} - \frac{m+1}{(p-1)^2} + \mathcal{O}(p^{-3})\Bigr) \Bigr].
  \end{align*}
  
  Now we can find a lower bound on the expected squared norm as
  \begin{align}\label{eqn:LBtilde}
    \mathbb{E}[ \| \beta- \tilde\beta^{\text{rs}} \|^2] &= \mathbb{E}[\sum_{j=1}^p \Big(\beta_j -  \tilde \beta_j^{\text{rs}}\Big)^2] =  \sum_{j=1}^p \beta_j^2 -2\beta_j\mathbb{E}[\tilde \beta_j^{\text{rs}}] + \mathbb{E}[(\tilde \beta_j^{\text{rs}})^2] \geq \\
    &\geq  \|\beta\|^2\cdot\Big(1-\frac{2m}{p} \Big) + \tau^2 ma(\frac{1}{p-1}-\frac{m+1}{(p-1)^2}+ \mathcal{O}(p^{-3})).
  \end{align}
  
  \textit{Upper bound for true coefficient:}
  
  The additional assumption on the diagonal elements proportional to the true coefficient ensures that $\Phi_{\text{pt}}$ has full row-rank. 
  From Lemma~\ref{lemma:ProjPhi}, we see that $\tilde \beta^{\text{pt}} = P_\Phi^{\text{pt}}\beta$ for $P_\Phi^{\text{pt}} = \Phi_{\text{pt}}'(\Phi_{\text{pt}}\Phi_{\text{pt}}')^{-1}\Phi_{\text{pt}}$ still
  equals
  \begin{align*}
    \tilde \beta_j^{\text{pt}} = \begin{cases} c\beta_j \cdot \frac{\sum_{k:h_k=h_j}(c\beta_k)\beta_k}{\sum_{k:h_k=h_j}c^2\beta_k^2} = \beta_j & \beta_j\neq 0 \\
    0 \cdot \frac{\sum_{k:h_k=h_j}(c\beta_k)\beta_k}{\sum_{k:h_k=h_j}c^2\beta_k^2} = 0 & \beta_j=0, \exists k\in h^{-1}(h_j): \beta_k\neq 0 \\
     d_j \cdot \frac{\sum_{k:h_k=h_j}d_k\overbrace{\beta_k}^{=0}}{\sum_{k:h_k=h_j}d_k^2} = 0 & \beta_j=0, \forall k\in h^{-1}(h_j): \beta_k=0
    \end{cases},
  \end{align*}
  the true coefficient $\beta$ in every case, implying $\beta= P_\Phi^{\text{pt}}\beta \in\text{span}(\Phi_{\text{pt}}')$. As a short remark, here 
  we see that the choice of diagonal elements $\{d_k:k\in h^{-1}(h_j)\}$ in the third case have no influence on the projection, as long as at least one is non-zero.
  
  Similarly to before, we need to bound the expectation of $(\beta-P_X^{\text{pt}}\beta)'\Sigma(\beta-P_X^{\text{pt}}\beta)$, 
  where $P_X^{\text{pt}} = \Phi_{\text{pt}}'(\Phi_{\text{pt}} X'X\Phi_{\text{pt}}')^{-1} \Phi_{\text{pt}} X'X$. Since $\beta\in\text{span}(\Phi_{\text{pt}}')$, 
  we have $\beta=P_X^{\text{pt}} \beta$ and, therefore,
  \begin{align}\label{eqn:UB_PrTh1}
    \mathbb{E}[(\beta-P_X^{\text{pt}}\beta)'\Sigma(\beta-P_X^{\text{pt}}\beta)]=0.
  \end{align}
  Finally, we can put the results together to obtain 
  \begin{align*}
    \mathbb{E}[(\tilde y - \hat y_{\text{rs}})^2] - \mathbb{E}[(\tilde y - \hat y_{\text{pt}})^2] &= \mathbb{E}[(\beta-P_{\text{rs}}\beta)'\Sigma(\beta-P_{\text{rs}}\beta)] -  \\
     & \qquad \mathbb{E}[(\beta-P_{\text{pt}}\beta)'\Sigma(\beta-P_{\text{pt}}\beta)] \geq \\
    &\geq \|\beta\|^2 \lambda_p(1-\frac{2m}{p}) + \frac{a}{p-1}m\lambda_p \tau^2  (1-\frac{m+1}{p-1} + \mathcal{O}(p^{-2})) .
  \end{align*}

\end{proof}
\begin{remark}\label{rem:ProofTh1}
  \begin{itemize}
    \item When using diagonal elements just almost proportional to the true $\beta$, we can obtain the upper bound
    \begin{align}
      \mathbb{E}[(\beta-P_X^{\text{pt}}\beta)'\Sigma(\beta-P_X^{\text{pt}}\beta)]\leq\lambda_1 \cdot \mathbb{E}\Big[\|\beta-\tilde \beta^{\text{pt}}\|^2 \cdot \Bigl(1+\|P_X^{\text{pt}}\|^2\Bigr)\Big],
    \end{align}
    where $\|P_X^{\text{pt}}\|$ is the spectral norm induced by the euclidean norm growing bigger when $X'X$ is further away from the identity.
    As long as $\|\beta-\tilde \beta^{\text{pt}}\|^2$ is small enough such that this upper bound remains smaller than the obtained lower bound for random sign diagonal elements,
     we still have a theoretical guarantee for an average gain in prediction performance.
    \item We assumed $\mathbb{E}[y_i]=0,\mathbb{E}[x_i]=0$ for notational convenience in the proof. With a general center $\mathbb{E}[x_i]=\mu_x$ and intercept $\mu\neq 0$ as in \eqref{eqn:datastr}, we can
    just use the centered $X$ and $y$ and the proof will work in a similar way for the same bound, but also needs to consider the estimation of the intercept $\hat \mu = \bar y - (\Phi\bar x)'(Z'Z)^{-1}Z'y$ for both $Z=Z_{\text{rs}},Z_{\text{pt}}$ and $\Phi=\Phi_{\text{rs}},\Phi_{\text{pt}}$.
    \item The assumption of multivariate normal distribution for the predictors allows us to explicitly calculate
    $\mathbb{E}[(\Phi X'X\Phi')^{-1}|\Phi]\Phi\Sigma\Phi'$ from the Inverse-Wishart-distribution, but we could also allow any distribution, for which this expression 
    does not depend on the choice of $\Phi$.
    \item  In the proof, we can see that the concrete adaption of diagonal elements to retain $\text{rank}(\Phi_{\text{pt}})=m$ is irrelevant, as long as there is 
    at least one non-zero $d_j$ with $j\in h^{-1}(i)$ for each $i\in [m]$. Our proposed adaption aims at adding minimal noise when we
    can not choose the diagonal elements exactly proportional to the true $\beta$ (e.g., when we only use the sign information), while keeping $\Phi_{\text{rs}}$ not just full rank but also well-conditioned.
  \end{itemize}
\end{remark}

\clearpage

\section{Additional Figures for Simulation Study}\label{sec:AppB}
 
Here, we include the additional figures for the simulation results mentioned and explained in Section~\ref{sec:simulation:res}.

\begin{figure*}
\centering
    \includegraphics[width=0.8\textwidth,trim=0 20 0 0, clip]{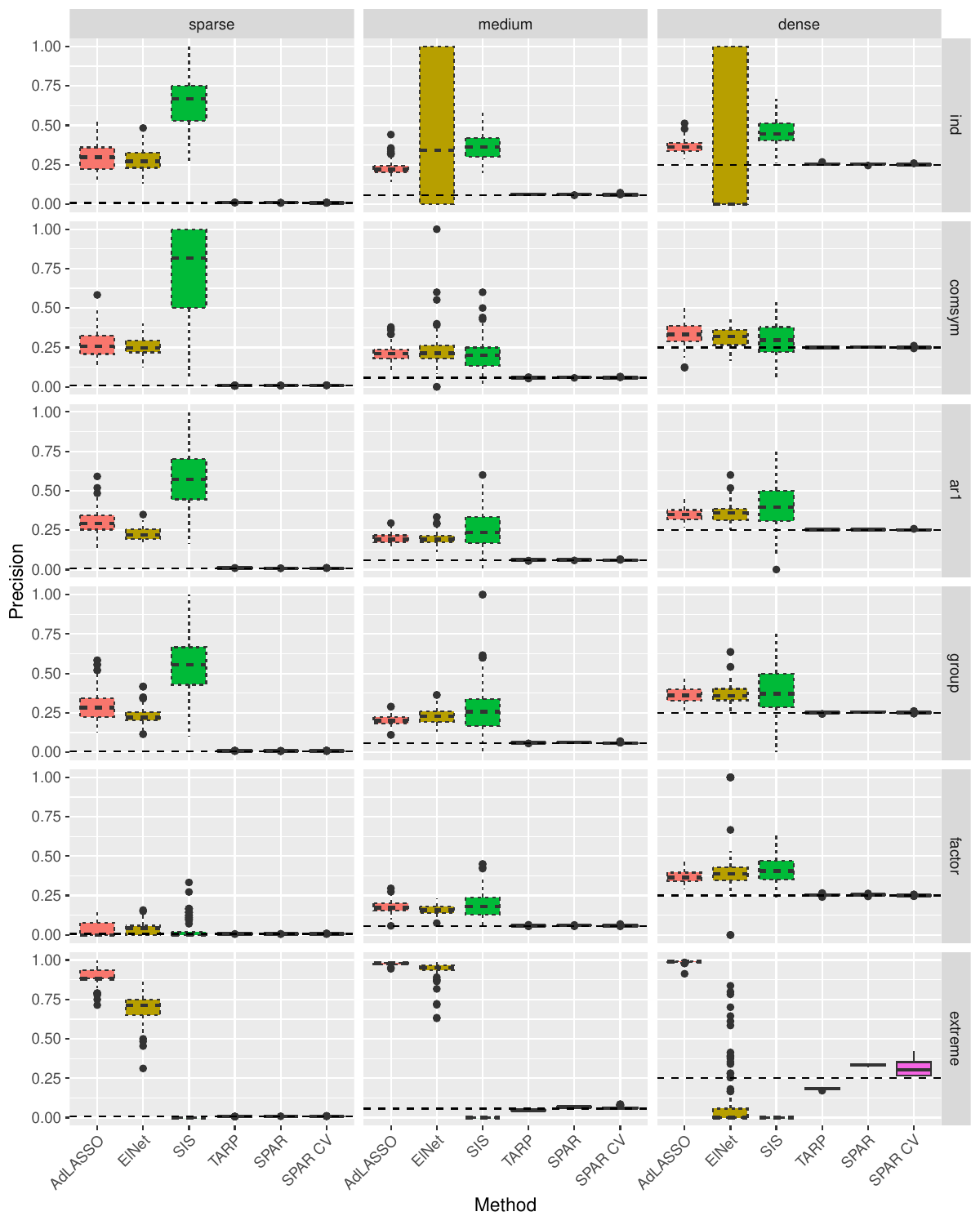}
  \caption{Precision of competing methods which perform any type of variable selection for different covariance and active predictor settings for $n_\text{rep}=100$ replications ($n=200,p=2000,\rho_\text{snr}=10$). Sparse methods are marked by dotted boxes.}  \label{fig:Prec_cov_settings}
\end{figure*}

 \begin{figure*}
   \centering
     \includegraphics[width=0.8\textwidth,trim=0 20 0 0, clip]{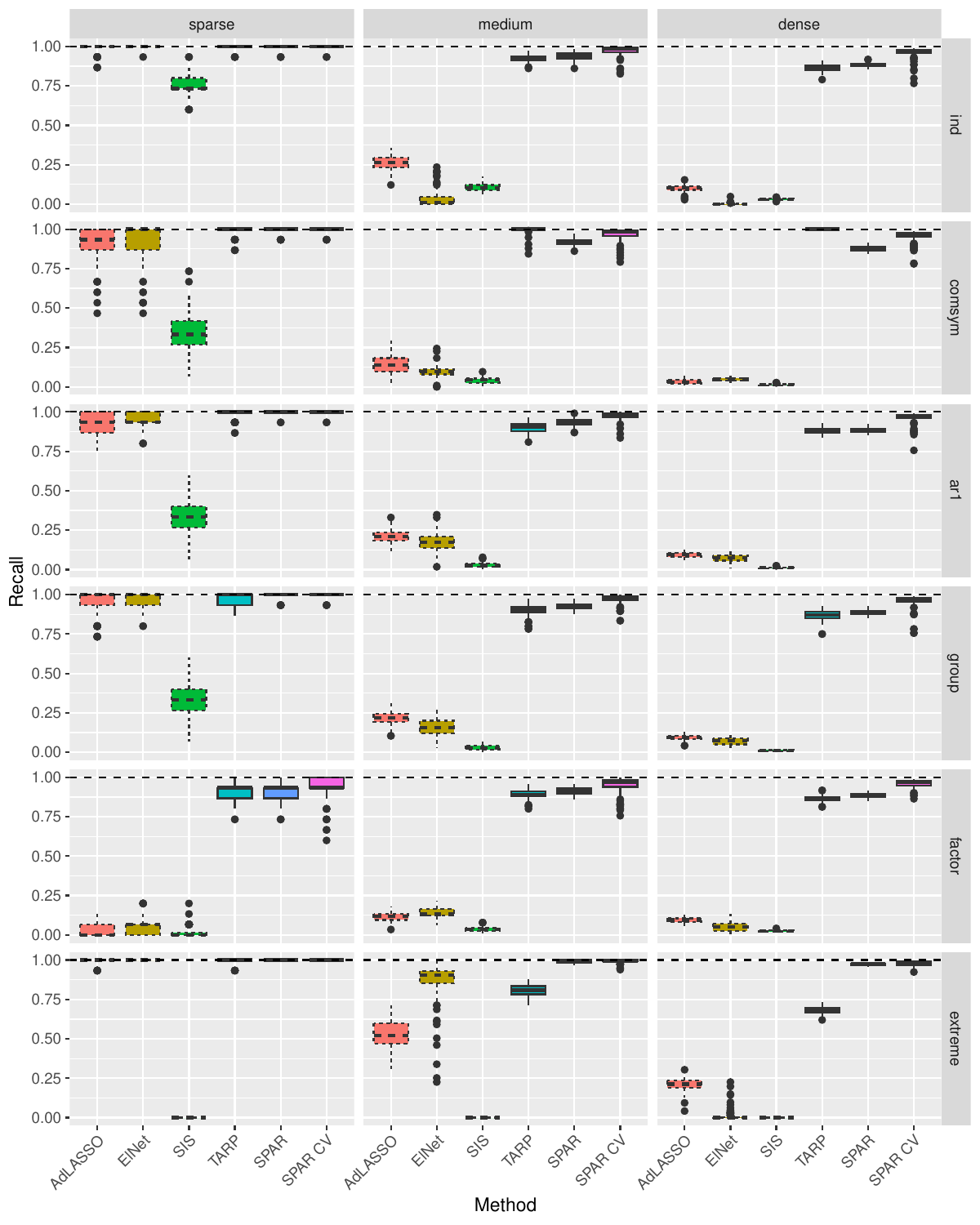}
   \caption{Recall of competing methods which perform any type of variable selection for different covariance and 
   active predictor settings for $n_\text{rep}=100$  replications ($n=200,p=2000,\rho_\text{snr}=10$). Sparse methods are marked by dotted boxes.}
   \label{fig:Rec_cov_settings}
 \end{figure*}

% \begin{figure*}
%   \centering
% \includegraphics[width=0.8\textwidth,trim=0 20 0 0, clip]{./PlotsPaper2024/SPAR_CV_AUC_cov_settings.pdf}
%   \caption{AUC over active and cov settings}
 %  \label{fig:AUC_cov_settings}
% \end{figure*}

\begin{figure*}
  \centering
    \includegraphics[width=0.8\textwidth,trim=0 20 0 0, clip]{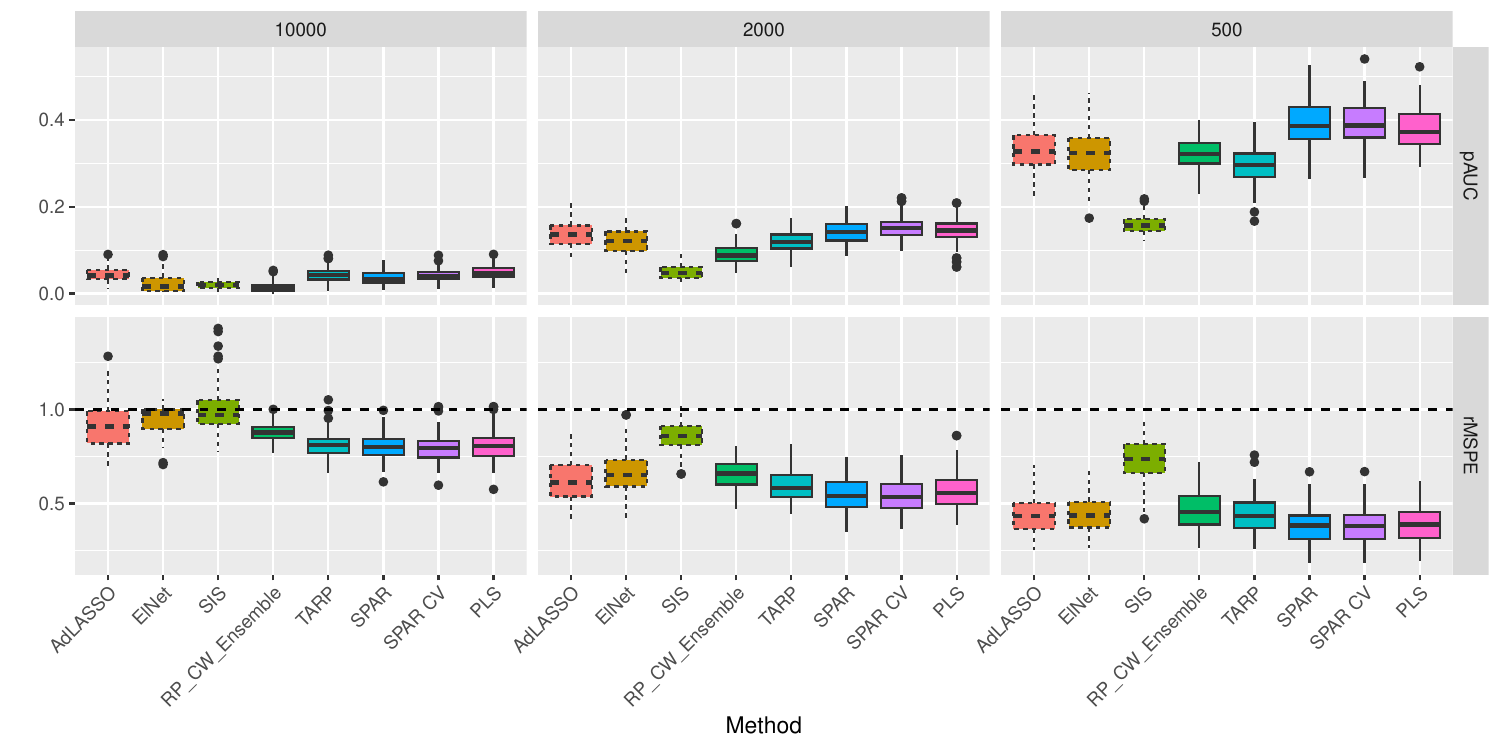}
       \includegraphics[width=0.8\textwidth,trim=0 20 0 0, clip]{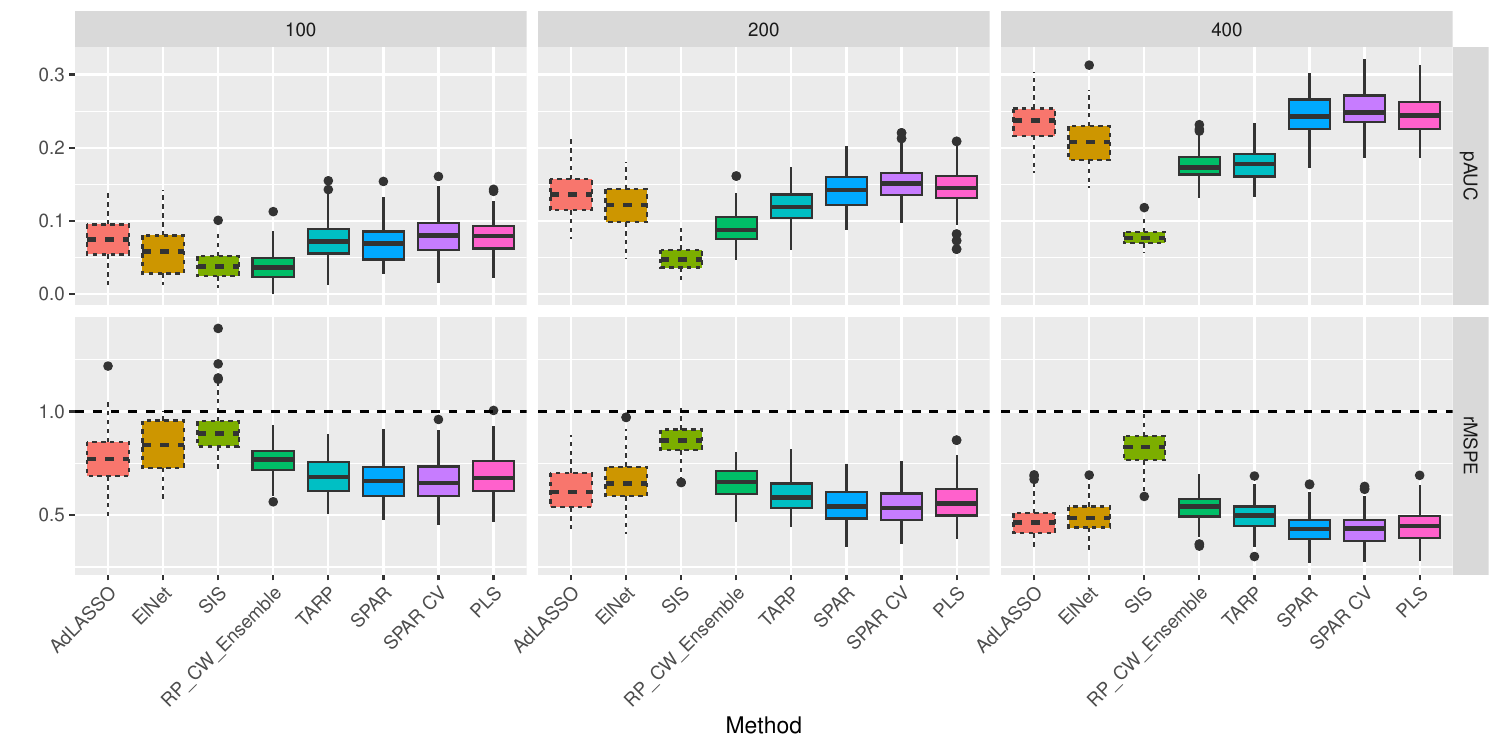}
    \includegraphics[width=0.8\textwidth,trim=0 20 0 0, clip]{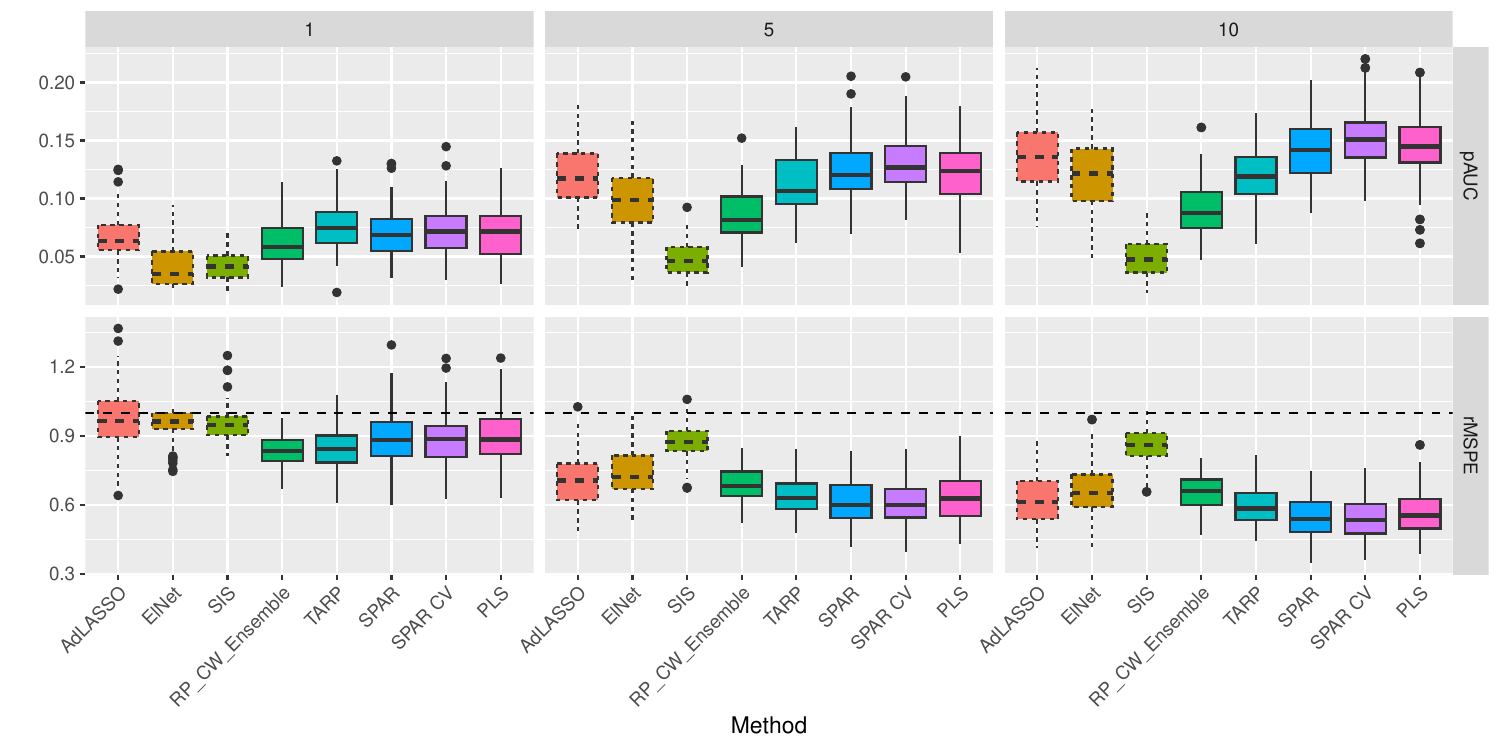}
 
  \caption{Performance measures of the competing methods, where the sparse methods are marked by dotted boxes, for `group' covariance setting, medium 
  setting for the active variables and $p=500,2000,10000$ (top panel, $n=200,\rho_\text{snr}=10$), $n=100,200,400$ (middle panel, $p=2000,\rho_\text{snr}=10$) and $\rho_\text{snr}=1,5,10$ (bottom panel, $p=2000,\rho_\text{snr}=10$) for $n_\text{rep}=100$ replications.}
  \label{fig:perf_group_medium}
\end{figure*}

% \begin{figure*}
%   \centering
%   \caption{Performance measures of the competing methods, where the sparse methods are marked by dotted boxes, for `group' covariance setting, medium 
%  setting for the active variables and $n=100,200,400$ for $n_\text{rep}=100$ replications ($p=2000,\rho_\text{snr}=10$)}
%  \label{fig:perf_group_medium_incN}
% \end{figure*}

% \begin{figure*}
%  \centering

%  \caption{Performance measures of the competing methods, where the sparse methods are marked by dotted boxes, for `group' covariance setting, medium 
%  setting for the active variables and $\rho_\text{snr}=1,5,10$ for $n_\text{rep}=100$ replications ($n=200,p=2000$)}
%  \label{fig:perf_group_medium_incSNR}
% \end{figure*}

\clearpage

%\section{Gene Selection in Section~\ref{sec:data:rat} }\label{sec:AppC}

% \rom{Remove this section?? when taking all $200$ as active our SPAR (and most other methods) always has $0$ pAUC}
\section{Simulation Study with Synthetic Data Based on Rateye Gene Expression Example}\label{sec:AppC}

We employ a simulation setting similar to the one in Section~\ref{sec:simulation} where we use the first $p=500, 2000, 22905$
genes from our filtered gene expression data in settings 1-5 of Section~\ref{sec:data:rat} as predictors, construct sparse, medium, and dense coefficient vectors $\beta$ as in Section~\ref{sec:simulation:datagen}, and generate a synthetic response with mean $\mu=1$ from the predictors with noise level chosen such that the signal-to-noise ratio based on the empirical predictor covariance is $10$.

Figure~\ref{fig:rateye_synth_rMSPE} shows the relative MSPE and Figure~\ref{fig:rateye_synth_rMSPE}  the partial AUC for sparse, medium, and dense settings and the different values of $p$ over $100$ replications. We observe that SPAR performs well in all settings, even in the sparse ones, especially for the case $p=22905$. We also provide the rank of the methods (from best to worst) in terms of prediction and variable ranking performance averaged over all scenarios and all repetitions in Table~\ref{tab:ranks_rateye}. Note that for PLS the implementation we used failed to work for $p=22905$ and we therefore report no values. 

We observe that SPAR CV achieves the best rank followed by SPAR and TARP in terms of rMSPE. SPAR CV and elastic net achieve a similar rank in terms of pAUC, followed by SPAR and AdLASSO. 

   \begin{figure}[h!]
    \centering
      \includegraphics[width=0.8\columnwidth,trim=0 21 0 0,clip]{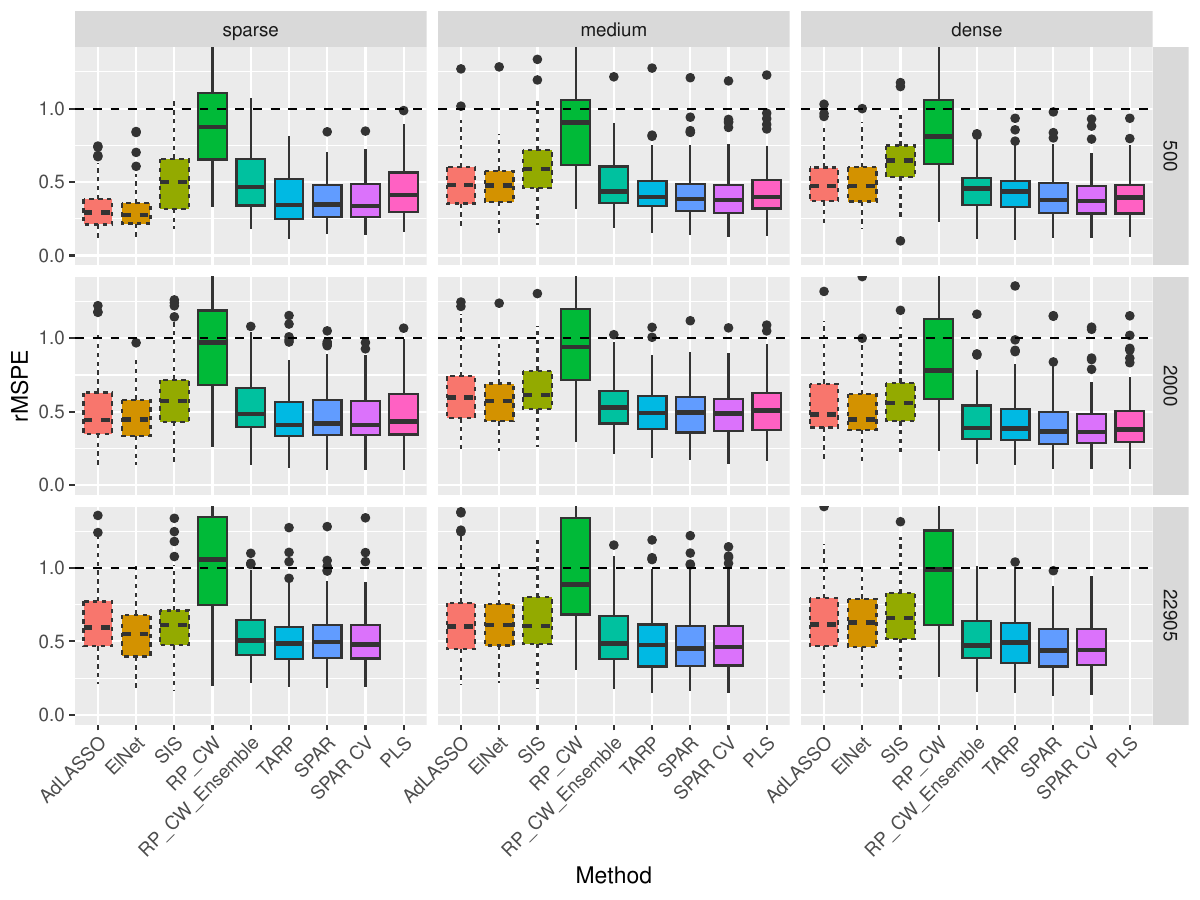}
   \includegraphics[width=0.8\columnwidth,trim=0 21 0 0,clip]{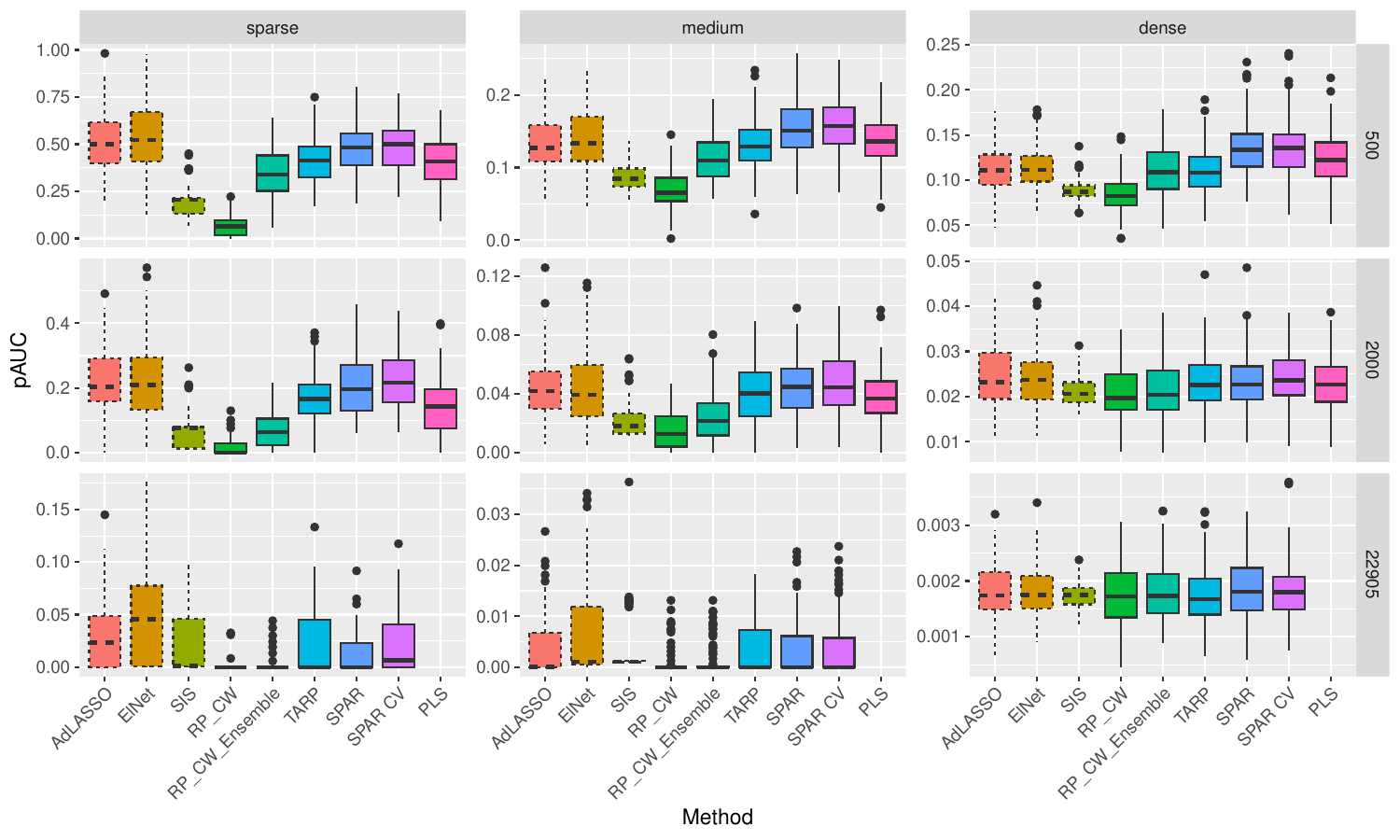}
   \caption{Relative MSPE and pAUC of the competing methods for different active predictor settings and different number of variables $p=500,2000,22905$ over $n_\text{rep}=100$ synthetic datasets using the observed genes in the gene expression dataset.}
    \label{fig:rateye_synth_rMSPE}
  \end{figure}

\begin{table}[t!]
    \caption{Mean and standard error of the rank (best to worst) achieved by each method in terms of rMSPE and pAUC across all investigated settings and $n_\text{rep}=100$ replications on the synthetic rat eye datasets. 
    The cells of the top 3 ranked methods are highlighted.}
    \label{tab:ranks_rateye}
    \centering
\begin{tabular}{lLlLl}
\toprule
Method & \multicolumn{2}{c}{rMSPE} & \multicolumn{2}{c}{pAUC}\\
\midrule
AdLASSO & 4.96 &(0.049) & 3.747 &(0.045)\\
ElNet & 4.449 &(0.049) & 3.321 &(0.048)\\
SIS & 5.933 &(0.042) & 5.348 &6(0.049)\\
RP\_CW & 7.63 &(0.023) & 6.542 &(0.045)\\
RP\_CW\_Ensemble & 4.237 &(0.038) & 5.482 &(0.044)\\
TARP & 3.238 &(0.034) & 4.419 &(0.047)\\
SPAR & 2.854 &(0.037) & 3.695 &(0.048)\\
SPAR CV & 2.699 &(0.036) &3.444 &(0.047) \\
\bottomrule
\end{tabular}

    \end{table}

\end{appendices}

\end{document}